\documentclass[twoside]{article}

\PassOptionsToPackage{}{natbib}

\usepackage[accepted]{aistats2024}
\usepackage[round]{natbib}

\usepackage[ruled,linesnumbered]{algorithm2e}
\usepackage{dsfont}
\usepackage{makecell}

\usepackage[utf8]{inputenc} 
\usepackage{hyperref}       
\usepackage{url}            
\usepackage{booktabs}       
\usepackage{amsfonts}       
\usepackage{nicefrac}       
\usepackage{microtype}      
\usepackage{xcolor}         
\usepackage{amsmath}
\usepackage{amssymb}
\usepackage{mathtools}
\usepackage{amsthm}
\usepackage{bbm}
\usepackage[capitalize,noabbrev]{cleveref}
\usepackage{graphicx}
\usepackage{aliascnt}
\usepackage{float}
\usepackage{xr-hyper}
\usepackage{hyperref}
\usepackage{todonotes}
\usepackage{varwidth}
\usepackage{multicol}
\usepackage{changepage}
\usepackage{enumitem}
\usepackage{autonum}
\usepackage{xargs}
\usepackage[ruled,linesnumbered]{algorithm2e}

\makeatletter
\newtheorem{theorem}{Theorem}
\crefname{theorem}{theorem}{Theorems}
\Crefname{Theorem}{Theorem}{Theorems}

\newaliascnt{lemma}{theorem}
\newtheorem{lemma}[lemma]{Lemma}
\aliascntresetthe{lemma}
\crefname{lemma}{lemma}{lemmas}
\Crefname{Lemma}{Lemma}{Lemmas}

\newaliascnt{corollary}{theorem}

\aliascntresetthe{corollary}
\crefname{corollary}{corollary}{corollaries}
\Crefname{Corollary}{Corollary}{Corollaries}

\newaliascnt{proposition}{theorem}
\newtheorem{proposition}[proposition]{Proposition}
\aliascntresetthe{proposition}
\crefname{proposition}{proposition}{propositions}
\Crefname{Proposition}{Proposition}{Propositions}

\newaliascnt{definition}{theorem}

\aliascntresetthe{definition}
\crefname{definition}{definition}{definitions}
\Crefname{Definition}{Definition}{Definitions}

\newtheorem{assumption}{\textbf{A}\hspace{-3pt}}
\Crefname{assumption}{\textbf{A}\hspace{-3pt}}{\textbf{A}\hspace{-3pt}}
\crefname{assumption}{\textbf{A}}{\textbf{A}}

\newaliascnt{remark}{theorem}
\newtheorem{remark}[remark]{Remark}
\aliascntresetthe{remark}
\crefname{remark}{remark}{remarks}
\Crefname{Remark}{Remark}{Remarks}

\crefname{example}{example}{examples}
\Crefname{Example}{Example}{Examples}

\crefname{algorithm}{algorithm}{algorithms}
\Crefname{Algorithm}{Algorithm}{Algorithms}

\crefname{figure}{figure}{figures}
\Crefname{Figure}{Figure}{Figures}

\usepackage{caption}
\usepackage{subcaption}
\usepackage{multirow}

\def\rset{\mathbb{R}}

\def\nset{\mathbb{N}}

\def\param{w}

\def\eqsp{\,}
\def\wrt{w.r.t.}

\def\Algo{\texttt{Generalized AsyncSGD}}
\def\iid{i.i.d.}

\def\nsteps{T}
\newcommandx{\M}[3][1=i,2=k,3=\nsteps]{\operatorname{M}^{#3}_{#1,#2}}
\newcommandx{\m}[3][1=i,2=k,3=\nsteps]{\operatorname{m}^{#3}_{#1,#2}}
\newcommandx{\mi}[1][1=i]{\operatorname{m}_{#1}}
\newcommandx{\weightm}[2][1=k,2=\nsteps]{\operatorname{m}^{#2}_{#1}}
\def\F{\mathcal{F}}

\newcommand{\ooint}[1]{\left(#1\right)}

\newcommand{\Set}{\mathcal{S}}

\newcommand{\norm}[1]{\left \lVert #1 \right \rVert}

\newcommand{\rmd}{\mathrm{d}}

\usepackage{xargs}
\newcommandx{\CPE}[3][1=]{\PE_{#1}\left[\left. #2 \, \right| #3 \right]}

\def\cardgrad{\operatorname{C}}

 \newcommand{\tcr}[1]{\textcolor{red}{#1}}
\def\PE{\mathbb{E}}
\def\PP{\mathbb{P}}

\newcommandx{\indi}[2][1=]{\1^{#1}_{#2}}
\newcommand{\indiacc}[1]{\1_{\{#1\}}}
\newcommand{\indin}[1]{\1\left\{#1\right\}}
\newcommand{\1}{\ensuremath{\mathbbm{1}}}

\begin{document}
\addtocontents{toc}{\protect\setcounter{tocdepth}{0}}

\twocolumn[

\aistatstitle{Queuing dynamics of asynchronous Federated Learning}

\aistatsauthor{ Louis Leconte \And Matthieu Jonckheere \And Sergey Samsonov \And Eric Moulines}

\aistatsaddress{Lisite, Isep, Sorbonne Univ.\,\,\\Math.\ and Algo. \\Sciences Lab, Huawei Tech.  \And  CNRS, LAAS, \\ and IRIT, France \And HSE University, \\Moscow, Russia \And CMAP\\ Ecole Polytechnique,\\ France} ]

\begin{abstract}
We study asynchronous federated learning mechanisms with nodes having potentially different computational speeds. In such an environment, each node is allowed to work on models with potential delays and contribute to updates to the central server at its own pace. Existing analyses of such algorithms typically depend on intractable quantities such as the maximum node delay and do not consider the underlying queuing dynamics of the system. In this paper, we propose a non-uniform sampling scheme for the central server that allows for lower delays with better complexity, taking into account the closed Jackson network structure of the associated computational graph. Our experiments clearly  show a significant improvement of our method over current state-of-the-art asynchronous algorithms on an image classification problem.
\end{abstract}

\section{Introduction}
Federated learning (FL) is a distributed learning paradigm that allows agents to learn a model without sharing data \citep{konevcny2015federated, mcmahan2017communication}. A central server (CS) coordinates the entire process. In most implementations, the CS uses synchronous operations. During each epoch, the CS communicates with a subset of clients and waits for their "local updates". The CS then uses these local updates to update the global model; \cite{mcmahan2017communication, wang2021high}. Nevertheless, different computational speeds, latencies, and/or transmission bandwidths lead to a cascade of issues such as delays and stragglers. In each epoch, CS  must keep up with the pace of the slowest agent.

A solution called \texttt{FedAsync} eliminates the structured rounds of CS interaction and transitions to asynchronous optimization \citep{xie2019asynchronous}. This approach, along with subsequent works in this direction \citep{chen2020asynchronous, chen2021towards, xu2021asynchronous}, enables asynchronous operation for the CS and agents. \texttt{FedAsync} facilitates the aggregation of agents updates through the CS, rendering the solution highly scalable. More recently, \cite{mishchenko2022asynchronous} has expanded the theoretical comprehension of purely asynchronous SGD within a homogeneous framework where all agents can access identical data;  certain limitations still persist in heterogeneous scenarios.

In practical applications of asynchronous federated learning (FL), interactions between agents and the CS require the use of queues for processing (potentially) multiple jobs. The distribution of processing delays varies significantly across agents, and this variability has been shown to have a negative impact on optimization processes.
 In this paper, we significantly improve the analysis delineated in \cite{koloskova2022sharper}, exploring in depth an asynchronous algorithm—\texttt{AsyncSGD}. This algorithm empowers nodes to queue tasks, initiating communication with the central server upon task completion. The subsequent analysis adheres to a virtual iterates sequence under standard non-convex assumptions. However, previous studies made overly simplistic assumptions about the dynamics of queues, choosing to represent them with an upper bound on the processing delays encountered by the CS. In contrast, our theory intricately models the queuing dynamics using a \textbf{stationary closed Jackson network}. This approach allows capturing precisely the queuing dynamics - number of buffered tasks, processing delay, etc$\dots$-, as a function of agents speed. We integrate assumptions about the service time distributions, enabling us to define the explicit stationary distribution of the number of in-service tasks. 

\paragraph{Contributions.} 
\begin{itemize}[leftmargin=*,itemsep=0pt]
    \item We identify key variables that affect the performance of the optimization procedure and depend on the queuing dynamics.
    \item Building on the findings of our analysis, we introduce a new algorithm called \Algo. This algorithm exploits non-uniform agent selection and offers two notable advantages: First, it guarantees unbiased gradient updates, and second, it improves convergence bounds.
    \item To gain deeper insights, we delve into the limit regimes characterized by large concurrency. In these contexts, our analysis shows that heterogeneity in server speeds can be balanced by the strategic use of non-uniform sampling among agents.
    \item Experimental results show that our approach outperforms other asynchronous baselines on a deep learning experiment.
    \end{itemize}

\paragraph{Related works}
\label{sec:related-works}
Up to this point, the focus has been on synchronous federated learning techniques, as evidenced by notable contributions such as \citep{wang2020tackling, qu2021feddq, makarenko2022adaptive, mao2022communication, tyurin2022dasha}. However, synchronous methods often suffer from suboptimal resource allocation and long training times. Moreover, as the number of participating agents grows, coordinating synchronous rounds with all participants becomes an increasingly difficult task for the central server (CS).

Synchronous federated learning methods are particularly vulnerable to the challenge of stragglers, prompting the emergence of research endeavors rooted in the principles of \texttt{FedAsync} and its subsequent extensions, as elucidated by \citep{xie2019asynchronous}. The core concept revolves around updating the global model upon receiving a local model at the central server (CS).  \texttt{ASO-Fed} \citep{chen2020asynchronous}  introduces memory-based mechanisms on the local client side.  \texttt{AsyncFedED} \citep{wang2022asyncfeded}, drawing inspiration from \texttt{FedAsync}'s instantaneous update strategy, proposes dynamic adjustments to the learning rate and the number of local epochs to mitigate staleness. 

Looking at the problem from a different perspective, \texttt{QuAFL} \citep{zakerinia2022quafl} introduces a concurrent algorithm that aligns closely with the \texttt{FedAvg} strategy. \texttt{QuAFL} seamlessly integrates asynchronous and compressed communication methods while ensuring convergence. In this approach, each client is allowed a maximum of $K$ local steps and can be interrupted. To address the variability in computational speeds across nodes, \texttt{FAVAS} (as discussed by \cite{leconte2023favas}) strikes a balance between the slower and faster clients.  

\texttt{FedBuff}~\citep{nguyen2022federated} addresses asynchrony and concurrency by incorporating a buffer on the server side. Clients conduct local iterations, with the CS updating the global model solely upon completion by a predefined number of different clients.

Similarly, the work presented by \citet{koloskova2022sharper} revolves around Asynchronous SGD (\texttt{AsyncSGD}), offering guarantees contingent on the maximum delay.  Recent developments by \cite{fraboni2023general} expand upon the ideas presented by \cite{koloskova2022sharper}, allowing multiple clients to contribute within a single round.  \citet{liu2021adaptive} diverges from the buffer-centric approach and develops Adaptive Asynchronous Federated Learning (\texttt{AAFL}) to address speed disparities among local devices. Similar to \texttt{FedBuff}, \citet{liu2021adaptive}'s method entails only a fraction of locally updated models contributing to the global model update. Convergence guarantees within asynchronous distributed frameworks commonly rely on an analysis contingent upon the maximum delay \citep{nguyen2022federated, toghani2022unbounded, koloskova2022sharper}, which substantially exceeds the average delay.

\cite{tyurin2023optimal} introduces a novel asynchronous algorithm, presenting optimal convergence guarantees under the assumption of fixed computational speeds among workers over time. Notably, \citet{mishchenko2022asynchronous} conducts an independent analysis of asynchronous stochastic gradient descent that does not rely on gradient delay. However, in the context of a heterogeneous (non-\iid) setting, convergence is guaranteed up to an additive term linked to the dissimilarity limit between the gradients of local and global objective functions.

AsGrad \citep{islamov2023asgrad} is a recent contribution that proposes a general analysis of asynchronous FL under bounded gradient assumption, and adapt random shuffling SGD to the asynchronous case. Most of standard asynchronous baselines can be expressed in the general form proposed in \cite{islamov2023asgrad}, and strong convergence guarantees are provided. But all derivations assume delays are finite quantities.

While an impressive body of research has been dedicated to establishing theoretical tools for the performance evaluation of communication networks, including the development of intricate scheduling mechanisms and models (as exemplified in \cite{Malek2022,Stav2018} and  references therein), the predominant focus has revolved around performance metrics such as delays/completion times (measured in terms of physical time per node), queue lengths, and throughputs. When it comes to modeling federated learning, the application of stochastic network paradigms is significant, based on a wealth of results in the field. However, it is important to recognize that the key metrics to be computed in FL are significantly different from traditional network metrics, as we will discuss in more detail shortly. In particular, the measurement of delay in this context must take into account server steps, which introduces a more complicated dependence on the dynamics of each queue within the network. Moreover, optimizing these novel metrics may require completely different resource allocation paradigms.

\section{Problem statement}
\label{sec:optim_bounds}
We consider the optimization problem:
\begin{equation}
\label{eq:optimization-problem}
\textstyle{\min_{\param \in \rset^d} \sum_{i=1}^{n} \PE_{(x, y) \sim D^{\operatorname{data }}_i}[\ell_{i}(\mathrm{NN}(x, \param), y)]}\eqsp.
\end{equation}
Here $d$ is the number of parameters (network weights and biases), $n$ is the total number of clients, $\ell_{i}$ is the local loss function, $\mathrm{NN}(x,\param)$ is the DNN prediction function,
$D^{\operatorname{data}}_i$ is the training data distribution on node $i$. In FL, the distributions $D^{\operatorname{data}}_i$ are allowed to differ between clients (statistical heterogeneity). Let us denote by
\begin{equation}
\textstyle{
f_i(w) := \PE_{(x, y) \sim D^{\operatorname{data }}_i}[\ell_{i}(\mathrm{NN}(x, \param), y)]
}
\end{equation}
the local function optimized on node $i$ and $f:= \frac{1}{n} \sum\nolimits_{i=1}^n f_i$. Each node $i$ does not compute the true gradient of the function $f_i$, but has access to a \emph{stochastic} version of the gradient, denoted by $\Tilde{g}_{i}$.

We consider the \emph{task} as a computation of a gradient (or possibly stochastic gradient) on one of the clients. We assume that $n$ clients process a fixed number of tasks $\cardgrad \in \nset^*$ in parallel, and the total number of CS steps is $\nsteps$. Once a task is completed by an agent, the corresponding update is sent to the CS, which updates the global model and then passes the updated parameters to a new agent with probabilities $(p_i)_{i=1}^n$. The selected agent might already be busy computing a previous update. When the agent is busy, the new job enters a queue that is serviced on a first-in-first-out basis (FIFO). To perform this analysis, we must establish the following definitions:
\begin{itemize}[leftmargin=*,itemsep=0pt]
\item  $J_k \in [1, n]$ is the node completing a task at the $k$-th CS epoch (or step),
\item   $K_{k+1} \in [1, n]$ is the node selected at step $k \in [1,\nsteps]$,
\item  $ X_{i,k}$ is the number of tasks in node $i$ at step $k$, $i \in [1,n]$, $k \in [1,\nsteps]$.
\end{itemize}
All these random variables can be constructed as deterministic functions of the \iid\ sequence $(R_l)_{l \in \mathbb N}$ which stands for the routing decisions and the sequences $(\xi_l^i)_{l \in \mathbb N,i=1, \ldots, n}$ which stand for the service times (durations) of the tasks in each node. These two \iid\ sequences are independent.

We denote by $\M$ the number of CS steps between the time that a task is sent to node $i$ and the time it is completed:
\begin{equation}
    \M[i][k][\nsteps]= \indiacc{i}(K_{k+1}) \sum\nolimits_{r=k}^{\nsteps} \indi{(\sum\nolimits_{l=k}^r\indi{J_l=i})<X_{i,k}}
\end{equation}
Finally, for $k \in \{0,\dots, \nsteps\}$ we define
\begin{equation}
\textstyle{I^T_k = \sum_{l=0}^k l \cdot \indiacc{k-l}(\M[K_{l+1}][\ell][\nsteps])}\eqsp,
\end{equation}
which is the CS step corresponding when the task was dispatched to node $J_k$.

Direct analysis of the server iterate $(\param_k)_{k\geq 0}$ is difficult because we do not have access to the joint distribution of $(J_k)_{k\geq 0}$, $(\M)_{k \ge 0}$ and $(I_k^T)_{k \ge 0}$. \cite{koloskova2022sharper} assumes an upper bound on the number of gradients that are pending at step $k$ but have not yet been applied. \cite{koloskova2022sharper} proposes to select new nodes with uniform probability. In \Algo\ (see \Cref{algo:new}), we add some degree of freedom by allowing the central server to select a new node $K_{k+1}$ with (possibly non-uniform) probability $\mathbf{p}= (p_j)_{j=1}^n$.
\begin{algorithm}[h]
\caption{\Algo \label{algo:new}}
\SetKwInOut{Input}{Input}
\SetKwInOut{Output}{Output}
\SetKwBlock{Loop}{Loop}{end}
\SetKwBlock{Initialize}{Initialize}{end}
\SetKwFor{When}{When}{do}{end}
\SetKwFunction{Wait}{Wait}
\SetKwFunction{ClientMain}{ClientLocalTraining}
\SetAlgoLined
\Input{Number of server steps $\nsteps$, Number of tasks $\cardgrad$ \;}
\tcc{\textbf{At the Central Server}}
\Initialize{
Initialize parameters $w_0$\;
Select initial set of clients $\Set_0$, with $| \Set_0 | =\cardgrad$ \;
Server sends $w_0$ to each client in $\Set_0$\;
All clients in $\Set_0$ compute gradients on $w_0$ \;
Compute optimal $(\mathbf{p},\eta)$ by minimizing \eqref{eq:optim} \;
}
\For{$k=0,\dots,\nsteps$}{
Server receives stochastic gradient $\Tilde{g}_{J_k}(\param_{I_k})$ \;
Update $\param_{{k+1}} \gets \param_{k} - \frac{\eta}{np_{J_k}} \Tilde{g}_{J_k}(\param_{I_k})$ \;
Sample a new client $K_{k+1}$ with prob. $p_{K_{k+1}}$\;
Send new model $\param_{{k+1}}$ to $K_{k+1}$ \;
}
\end{algorithm}

We denote
\[
\m := \PE[\M]\eqsp, \text{ and } \weightm := \sum_{i=1}^n \m/(n^2p_i^2)\eqsp.
\]
It is worth noting that $\m$ depends on the sampling probability $p_i$, but for simplicity, we prefer not to index explicitly by $\mathbf{p}:= (p_j)_{j=1}^n$.
We will show in \Cref{sec:closed_network} that, $\lim_{k \to \infty} \lim_{T \to \infty} \m = \mi $
(these expectations become stationary)
see \Cref{prop:mik}.
Of course, the exact analysis of the transient behavior is very complex:  simple upper bounds can be computed, but these are generally not expressive and hide the influence of key parameters.
\begin{figure}
    \centering
    \includegraphics[scale=.4]{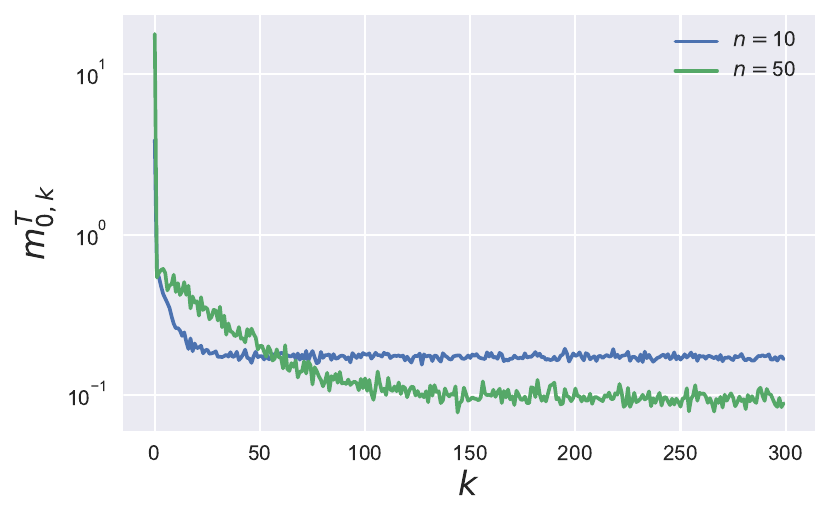}
    \caption{Evolution of $\m$ \wrt\ $k$, for two networks of size $n=10, 50$ initialized with full concurrency.}
    \label{fig:stationarity}
\end{figure}
In \Cref{fig:stationarity}, we simulate $n=10$ and $n=50$ nodes with $\cardgrad=n$ initial tasks. In this simulation, nodes $\{0, 1, 2, 3, 4\}$ are $10$ times faster than the other nodes to compute a task. Without loss of generality, we focus on the first node $(i=1)$, and we plot the value of $\m$ with respect to $k$, for $T=500$. The value of $\m$ becomes stationary when $k > 50$ and $k>150$, for $n=10$ and $n=50$, respectively.

In our analysis, in line with the approach presented in \cite{koloskova2022sharper} (but the same idea was applied  earlier), we introduce the  \emph{virtual iterates} $\mu_k$ as follows:
\begin{equation}
\begin{cases}
    \mu_0 = \param_0,\\
    \mu_1 = \mu_0 -\eta \sum_{i \in \Set_0} \frac{1}{np_i}\Tilde{g}_{i}(\param_{0}),\\
    \mu_{{k+1}} = \mu_{k} - \frac{\eta}{np_{K_k}} \Tilde{g}_{K_k}(\param_{k})\eqsp, \quad k \geq 1\eqsp.
\end{cases}
\end{equation}
In fact $\mu_k$ is defined as if the selected client $K_{k}$ was instantaneously contributing to the server update. Note that the gradients are computed on $\param_{k}$, not on $\mu_{k}$. The difference between $\mu_{k}$ and $\param_{k}$ consists of all the gradients computed (on potentially outdated $\param$'s) and not applied yet, $\mu_k-\param_k=\sum\nolimits_{(i,j) \in \mathcal{I}_k} \frac{-1}{np_i}\Tilde{g}_{i}(\param_{j})$, with $\mathcal{I}_k = \{ (i, j) \in [1,n] \times   [1,k] | (X_{i,k} \neq 0)  \text{ and }
 (\sum\nolimits_{i=1}^n \M[i][j][\nsteps] > k-j) \}$.

\section{Non-convex bounds}
\label{sec:nonconvexbounds}
Following the setting considered in \cite{koloskova2022sharper}, \cite{nguyen2022federated}, we focus on the scenario of the optimization problem \eqref{eq:optimization-problem} with $L$-smooth and nonconvex objective functions $f_i$. Proofs are detailed in \Cref{sec:app:proofs}. Our analysis is based on the following assumptions:
\begin{assumption}
\label{assum:uniflowerbound}
    Uniform Lower Bound: There exists $f_* \in \mathbb{R}$ such that $f(\param) \geq f_*$ for all $\param \in \mathbb{R}^d$.
\end{assumption}
\begin{assumption}
\label{assum:smoothgradients}
    Smooth Gradients: For any client $i$, the gradient $\nabla f_i$ is $L$-Lipschitz continuous for some $L>0$, i.e. for all $\param, \mu \in \mathbb{R}^d$:
$
\left\|\nabla f_i(\param)-\nabla f_i(\mu)\right\| \leq L\|\param-\mu\| .
$
\end{assumption}
\begin{assumption}
\label{assum:boundedvariance}
    Bounded Variance: For any client $i$, the variance of the stochastic gradients is bounded by some $\sigma^2>0$, i.e. for all $\param \in \mathbb{R}^d$:
$
\mathbb{E}[\left\|\widetilde{g}_i(\param)-\nabla f_i(\param)\right\|^2] \leq \sigma^2 .
$
\end{assumption}
\begin{assumption}
\label{assum:graddissim}
Bounded Gradient Dissimilarity: There exist constant $G$, such that for all $\param \in \mathbb{R}^d$:
$\left\|\nabla f_i(\param)-\nabla f(\param)\right\|^2 \leq G^2.$
\end{assumption}
The assumption \Cref{assum:boundedvariance} can be generalized to the \emph{strong growth condition} \cite{vaswani2019fast}:
\[
\textstyle{
\mathbb{E}[\left\|\widetilde{g}_i(\param)-\nabla f_i(\param)\right\|^2] \leq \sigma^2 + \rho^2 \left\| \nabla f_{i}(\param)\right\|^2
}\eqsp,
\]
following  \cite{beznosikov2023first}. Full details are given in the appendix.

\begin{theorem}
\label{theorem:cvgquant}
Assume \Cref{assum:uniflowerbound} to \Cref{assum:graddissim} and let the learning rate $\eta$ satisfy $\eta \leq \eta_{\max}(\mathbf{p})$, where
\begin{equation}
\label{eq:condition-learning-rate}
 \eta_{\max}(\mathbf{p}) =:\frac{1}{4L}\bigl(\cardgrad^{-1/2} \max_{k \leq \nsteps} \{ \weightm \}^{-1/2} \wedge 2/\sum\nolimits_{i=1}^n\frac{1}{n^2p_i}\bigr)\eqsp.
\end{equation}
Then \Algo\ converges at rate:
\begin{align}
\label{eq:cvg_theorem}
    \sum_{k=0}^{\nsteps} &\frac{\PE[\Vert \nabla f(w_k) \Vert ^2]}{8(\nsteps+1)}  \leq \frac{\PE[f(\mu_0) - f(\mu_{\nsteps+1})]}{\eta (\nsteps+1)} \\ \nonumber
&
+ \frac{\eta L B}{n} \sum\nolimits_{i=1}^n \frac{1}{n p_i} + \frac{\eta^2 L^2 B \cardgrad}{n} \sum\nolimits_{i=1}^n  \frac{\sum\nolimits_{k=0}^{\nsteps} m_{i,k}^\nsteps}{n p_i^2(\nsteps+1)} \eqsp,
\end{align}
where $B=2 G^2 + \sigma^2$.
\end{theorem}
The upper bound in \Cref{theorem:cvgquant} includes three distinct terms. The first term is a standard component associated with a general nonconvex objective function; it expresses how the choice of initialization affects the convergence process.

The second and third terms depend on the statistical heterogeneity within the client distributions and the fluctuations of the minibatch gradients. If we assume uniform probabilities, the second term agrees with that of \texttt{FedAvg}: this is the bound that would be obtained with synchronous optimization. In contrast, the third term encapsulates the unique challenge posed by optimization within an asynchronous framework.

Before moving on to the main study, it is important to analyze the behavior of this bound. Note first that the upper bound is minimized by $T \to \infty$ if we set $\eta = O( T^{-1/2})$. In this setting, the third term of the upper bound, which is proportional to $\eta^2$, becomes negligible. To obtain the optimal probability value $\mathbf{p}$, one should minimize $\sum_{i=1}^n 1/p_i$ in this regime, subject to the condition $\sum_{i=1}^n p_i = 1$. This minimization is achieved when $p_i =1/n$. Thus, with $T \to \infty$, a uniform distribution of weights turns out to be a reasonable choice.

\paragraph{A worked-out example} For regimes other than \( T \to \infty \), the bound given by Eq.~\eqref{eq:cvg_theorem} proves difficult to handle due to the complicated relationship between \( \m \) and the sampling distribution \( \mathbf{p} \). In the next section, we will apply queuing theory methods to shed light on these quantities. However, before diving into this detailed analysis, we will first examine the behavior of the bound using a simple example.
We choose the sampling probabilities $\mathbf{p}$ and the step size $\eta$ by solving the constrained optimization problem $\min_{\mathbf{p},\eta} G(\mathbf{p},\eta)$ as a function of $\eta \leq \eta_{\max}(\mathbf{p})$, where
\begin{multline}
G(\mathbf{p},\eta)= \frac{A}{\eta (\nsteps+1)} + \frac{\eta L B}{n}\sum\nolimits_{i=1}^n \frac{1}{np_i} \\ + \frac{\eta^2 L^2 B \cardgrad}{n} \sum\nolimits_{i=1}^n  \frac{\sum\nolimits_{k=0}^{\nsteps} \m}{n p_i^2(\nsteps+1)}\eqsp,
\label{eq:optim}
\end{multline}
and where $A=\PE[f(\mu_0) - f(\mu_{\nsteps+1})]$. To better understand this bound, let us examine a simple case. Suppose we have $n=100$ clients that are classified as either "fast" or "slow".  There are $n_f=90$ fast clients and $n-n_f=10$ slow clients, which are assumed to have the same characteristic (within each group).  We will focus on how the proposed bound behaves based on the ratio of the processing speed of the "fast" and "slow" clients, the proportions of fast and slow clients, and the concurrency. We will also examine two situations: one in which the processing time for gradient requests is fixed, and another in which it follows an exponential distribution. By default, slow clients process a gradient in a time unit of 1 (on average in the random case), while fast clients take $\frac{1}{\mu_f} \leq 1$ units on average. Let $p \in \ooint{0,1}$. We denote $p$ the probability to select one of the  fast clients, and $q=\frac{1}{n-n_f}-p\frac{n_f}{n-n_f}$ the probability of selecting one of the slow clients (we need $n_fp+(n-n_f)q=1$). The parameters are $L = 1$, $B = 20$ (to assess the effects of gradient noise and statistical heterogeneity), $A = 100$ (to highlight the impact of initial conditions). The number of tasks is varied $\cardgrad=10,50,100$, to assess the impact of concurrency. For the total number of CS iterations, we consider $\nsteps = 10^4$. We plot the selection probability $p$ versus $\mu_f$, ranging from $2$ to $16$ in \Cref{fig:proba_optimal}.
\begin{figure}
    \centering
    \includegraphics[scale=.4]{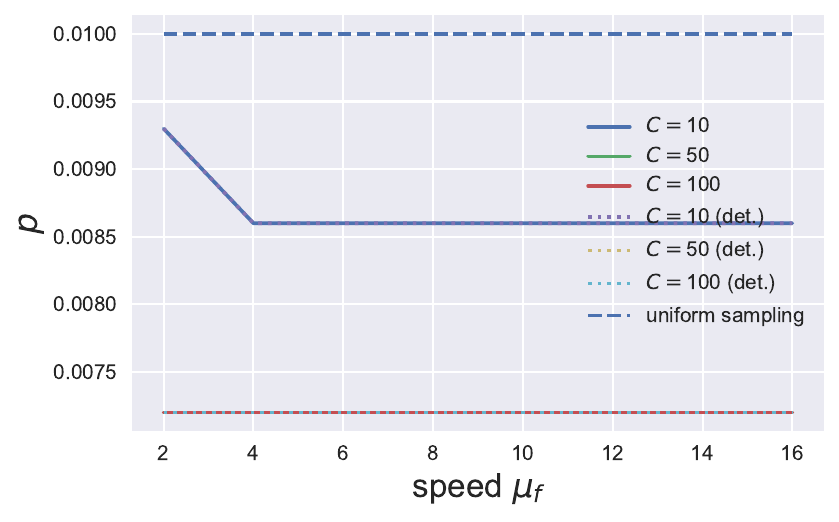}
    \caption{Optimal sampling probability $p$ as a function of the speed for different concurrency levels. The number of nodes is fixed to $n=100$ nodes.}
    \label{fig:proba_optimal}
\end{figure}
Furthermore, we graphically illustrate the relative improvement of the upper bound when compared with the uniform selection problem in \Cref{fig:relative_bounds}.
\begin{figure}
    \centering
    \includegraphics[scale=.4]{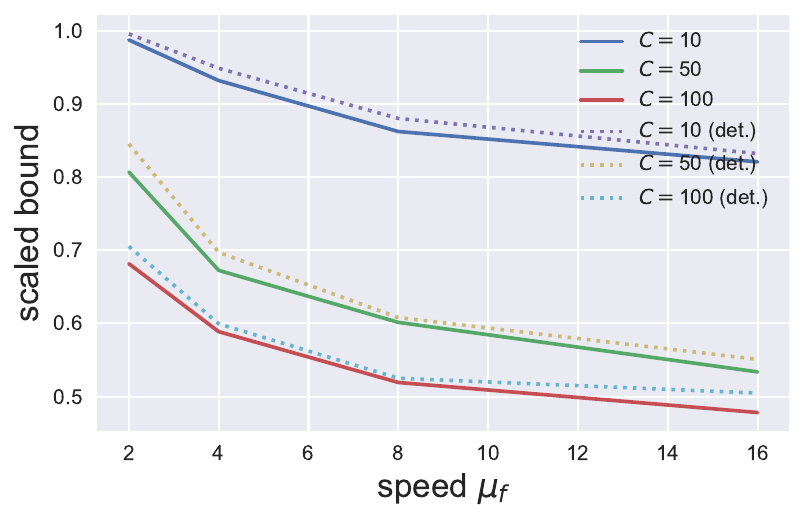}
    \caption{Relative improvements of the upper bounds as a function of the speed for different concurrency levels. The number of nodes is fixed to $n=100$ nodes.}
    \label{fig:relative_bounds}
\end{figure}
The results show that a significant improvement may be achieved (from 30\% when $\mu_f=2$ to 55\% when $\mu_f=16$). To achieve this improvement, we should decrease the probability of selecting fast clients to \(p= 7.3  \cdot 10^{-3} \). The conclusion (that will be justified theoretically in the next section) is that fast agents should be selected less frequently than slow agents. Even though this result may appear to be counter intuitive, it is justified by the fact that by selecting slow customers more frequently, processing times are reduced. Our simulations shows that the average delay at the CS is divided by $10$ and $2$, for fast and slow nodes respectively. More details are given in \Cref{sec:app:2clusters}.

Finally, these experiments also show that the distribution of the working time required for gradient evaluation does not have a significant impact: results are very similar whether the working time is deterministic or distributed according to an exponential (and therefore random) distribution (provided that the mean are preserved).

\paragraph{Comparison with \texttt{FedBuff} and \texttt{AsyncSGD}}
We emphasize that previous analyses of both \texttt{FedBuf} and \texttt{AsyncSGD} are based on strong assumptions: the queuing process is not considered in their analysis. In practice, slow clients with delayed information contribute. \cite{nguyen2022federated,koloskova2022sharper} propose to bound this delay uniformly by a quantity $\tau_{max}$. We retain this notation while reporting complexity bounds in \Cref{tab:cvgcomparison}, but argue that nothing guarantees that $\tau_{max}$ is properly defined.
\begin{table*}
    \centering
    \caption{Asynchronous bounds (up to numerical constants) under non-convex assumption for $\nsteps$ server steps. $\cardgrad$ is the number of initial tasks. $A=\PE[f(\mu_0) - f_*]$, and $B = {2G^2+\sigma^2}$. \tcr{$\tau_{\max}$} is defined in \cite{toghani2022unbounded} as the maximum delay. $\textcolor{magenta}{\tau_{c}}, \tau_{\text{sum}}^i$ are defined in \cite{koloskova2022sharper} as the average number of active nodes, and the sum of delays of node $i$, respectively.}
\resizebox{\linewidth}{!}{
\begin{tabular}{l|c|c}
\toprule
Method & Bounds & $\eta$\\
         \midrule
        \texttt{FedBuff} & $\frac{A}{\eta (\nsteps + 1)} + {\eta} LB + \eta^2 \tcr{{\tau^2_{\max}}} L^2Bn$ & $\leq \frac{1}{L\sqrt{\tcr{{\tau^3_{\max}}}}}$\\
        \texttt{AsyncSGD} & $\frac{A}{\eta (\nsteps+1)} + {\eta} LB + \eta^2 \textcolor{magenta}{\tau_{c}} {L^2B} \sum\nolimits_{i=1}^n  \frac{\tau_{\text{sum}}^i}{\nsteps+1} $ & $\leq \frac{1}{L\sqrt{\tau_c\tcr{\tau_{\max}}}}$\\
        \thead{\texttt{Generalized} \\ \texttt{AsyncSGD}} & $\frac{A}{\eta (\nsteps + 1)} + {\eta}LB\sum\nolimits_{i=1}^n \frac{1}{n^2p_i} + \eta^2 \cardgrad L^2B \sum\nolimits_{i=1}^n  \frac{\sum\nolimits_{k=0}^{\nsteps} \m}{n^2p_i^2(\nsteps+1)} $ & $\leq \frac{1}{L\sqrt{ \cardgrad \max_{k \leq \nsteps}\weightm}}$ \\
        \bottomrule
    \end{tabular}
     }
    \label{tab:cvgcomparison}

\end{table*}
In \Cref{fig:bounds_comparison} we compared the relative improvements of the upper bounds obtained with \Algo, \wrt\ \texttt{FedBuff} and \texttt{AsyncSGD} for the scenario described in the previous paragraph. The plot illustrates the massive improvement achieved by \Algo\ when optimal selection probabilities are used. It also illustrates that the bounds previously reported in the literature do not capture the essence of the problem. In particular, this comparison holds under the condition that the work time for gradient evaluation is deterministic, such that \( \tau_{\max} \) equals \( C \) times the work time of a slow client. When the working time is exponential, the maximum delay as defined in the analyses of \texttt{FedBuff} and \texttt{AsyncSGD} is infinite, and the bounds in \cite{nguyen2022federated} and \cite{koloskova2022sharper} are then empty.
\begin{figure}
    \centering
    \includegraphics[scale=.4]{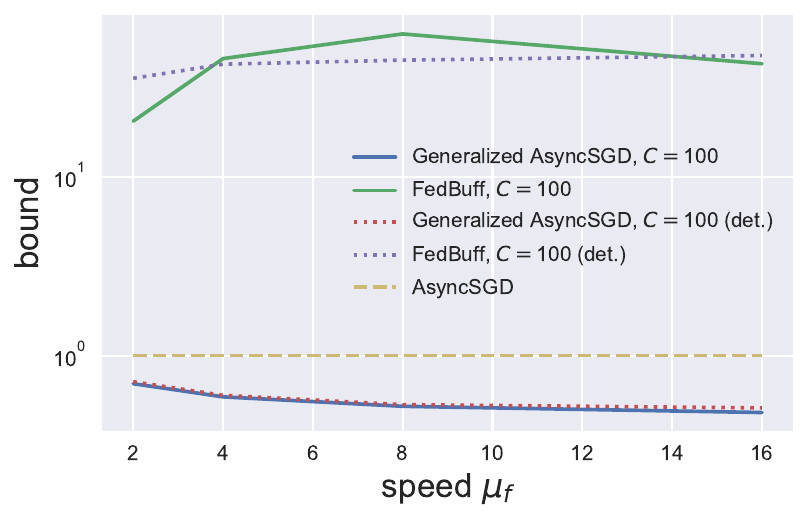}
    \caption{Relative improvement of \Algo\ over \texttt{FedBuff} and \texttt{AsyncSGD} as a function of  speed. The number of nodes is fixed to $n=100$ nodes.}
    \label{fig:bounds_comparison}
\end{figure}

\paragraph{Physical time \wrt\ CS number of epochs} The complexity bounds of \Algo\ are based on the number of communication rounds. The conclusions would, of course, be different if we took physical time as the criterion. Indeed, when we determine complexity in terms of number of communications, we do not take into account the time intervals between two successive arrivals at the central server (whose law depends on the relative speeds of the agents and the weights $\mathbf{p}$). We discuss bounds for \Algo\ \wrt\ physical time in \Cref{sec:app:physical_time} and assess them.

\section{Closed network}\label{sec:Jack}
\label{sec:closed_network}
The aim of this section is to obtain theoretical guarantees using precise results on closed Jackson networks \citep{jackson1954queueing, jackson1957networks}. We assume in this section that task duration follows an exponential distribution, with each user having its own mean. While it is feasible to extend these findings to deterministic durations and even almost arbitrary duration distributions, this complicates the theory. This approach not only captures the different speeds of the clients, but also allows for a precise assessment of the system dynamics. Consequently, we can accurately determine the critical values $\m$ in a steady state. Detailed proofs are postponed to the \Cref{sec:app_proof_queuing}.
\paragraph{Stationary distribution and key performance indicators}

Let us denote by $(D_i(t))_{i=1,\ldots,n}$ the number of task departures from node $i$ at time $t$, with the convention that $D_i(0)=0$.
$(T_{i,l})_{l \in \mathbb N}$ are the jump
times associated with the counting process $D_i$.
We further denote by $N(t)$ the number of tasks arriving at the CS, given by
$$
\textstyle{N(t)= \sum_{i=1}^n D_i(t)}\eqsp,
$$
while $(T_l)_{l \in \mathbb N}$ are the jump times associated with $N$. Observe that the indices $k$ used in the previous Section correspond to those jump times.
Finally we define the sequences of times $(\tau_{i,l})_{l \in \mathbb N}$ as the arrival times to node $i$ after time $0$.

In what follows, we assume that the task duration is \iid\ and exponentially distributed at rate $\mu_i$, and that the routing decisions are also \iid\ (and independent of everything else).
As in the previous section, we denote by $p_i$ the probability that the dispatcher sends a task to node $i$.

We denote by $X(t)=(X_{1}(t), \ldots, X_n(t))$ the continuous-time stochastic process describing the number of tasks in each node. The unit vectors in $\mathbb N^n$ are denoted by $(e_i)_{i=1\ldots,n}$.
We have the following results for $X$.
\begin{proposition}\label{prop:Jack}
Under the above assumptions, the dynamics of $(X(t), t\geq 0)$ is that of a closed Jackson network on the complete graph with $n$ nodes and $\cardgrad$ tasks.
The generator of the corresponding jump Markov process is given for all $x \in \mathbb N^n, i \in \mathbb N, j \in \mathbb N$ by
\[
q(x,x+e_i-e_j) = p_i \mu_j \mathds{1}(x_j >0)\eqsp.
\]
Furthermore, defining $\theta_i=\frac{p_i}{\mu_i}$, the stationary distribution of $X$ may be expressed as:
\begin{equation}
\textstyle{
\pi_C(x_1, \dots, x_{n}) = H_{\cardgrad}^{-1} \prod\nolimits_{i=1}^{n} \theta_i^{x_i}}\eqsp,
\end{equation}
with
$H_{\cardgrad} = \sum_{x: \sum_i x_i = \cardgrad } \prod \nolimits_{i=1}^{n} \theta_i^{x_i}
$.
\end{proposition}

Building on the understanding gained in \Cref{sec:optim_bounds}, our goal is to quantify the number of server steps that are executed when a new task arrives at a given node (say node $i$)
 and subsequently returns to the dispatcher.

For simplicity, the analysis is performed in the stationary regime. In particular, this means that at time $0$ the distribution of the number of tasks in each node follows the product measure defined in \Cref{prop:Jack}. We denote by $\PE^{\cardgrad}$ the stationary average of the closed Jackson network when the total number of tasks is equal to $\cardgrad$.
We can now state the main result of this section:
\begin{proposition}\label{prop:mik}
Given the model assumptions and assuming stationarity, for all $k \in \nset$,
\begin{equation}
\textstyle{
\lim_{T \to \infty} \m
 =  \PE^{{\cardgrad-1}}\Big[\int_0^{S_i} \sum\nolimits_{j=1}^{n}\mu_j \mathds{1}(X_j(s)>0) \,\rmd s\Big]
}\eqsp.
\end{equation}
\end{proposition}
From now on, we use the notation $\mi {=} \lim_{T \to \infty} \m.$
This quantity is in general difficult to simplify further.
We consider in the sequel a specific regime in which we can obtain tractable expressions as the Jackson network gets close to saturation. We now describe how the queue length $X_i$, and $\mi$ depend on the agents speed and selection probability $\mathbf{p}$, under this saturated stationary regime similar to the Halfin-Whitt regime in queuing theory \citep{halfin1981heavy}.
\paragraph{Scaling regime}
We rely on scaling bounds to provide rules of thumb when certain traffic conditions are satisfied.
We follow the derivation of \cite{van2021scaling}, which considers closed Jackson networks under a particular load regime. We assume without loss of generality that $\theta_{n} = \max_{i \in [1, n]}(\theta_i)$; where $\theta_i = p_i/\mu_i$ (see \Cref{prop:Jack}). Due to the closed nature of the network, rescaling all  parameters through division by the maximum traffic load leads to a different normalization constant $\tilde H_C$, but otherwise has no effect on the stationary joint distribution:
\begin{equation}
\textstyle{
    \pi(x_1, \dots, x_{n-1}) = \tilde H_C^{-1} \prod\nolimits_{i=1}^{n-1} {\gamma_i}^{-x_i}
}\eqsp,
\end{equation}
where for all $i \in [1,n]$, $\gamma_i = \theta_{n} / \theta_{i}$.
We consider a scaling regime where all nodes are saturated, but at different rates. In \Cref{sec:app:3clusters} we also define a more complicated scenario where some queue lengths may degenerate to $0$.
\paragraph{2 clusters under saturation}
We consider two clusters of nodes of size $n_f$ and $n-n_f$, respectively. Nodes $i \in [1,n_f]$ are fast, the rest are slow.
We assume that nodes from the same cluster have the same speed $\mu_f, \mu_s$, for fast and slow nodes, respectively ($\theta_f  < \theta_s$).
This gives the \emph{scaled} intensity of the slow nodes $\gamma_s(\iota)=1$, and the fast nodes $\gamma_f(\iota):=\frac{\theta_s}{\theta_f}$, where $\iota$ is the scaling parameter and the scaling regime corresponds to choosing those values as $\gamma_f(\iota) = 1+ c_f \iota^{\alpha-1}$; with $c_f>0$ a fixed positive constant, and $\alpha \leq 1$, while the total number of tasks also scales as follows:
$$
\textstyle{
\beta \iota^{1-\alpha} = \cardgrad+1
}\eqsp.$$
Choosing $\alpha \leq 1$ as in \cite{van2021scaling} ensures that node  loads approach 1 as $\iota \rightarrow \infty$, enabling the application of Corollary 2 from \cite{van2021scaling}. This yields precise results on saturated node queue lengths at high traffic loads, while queue lengths for the remaining nodes are determined by population size constraints.
In this context, define $X_{i}^\iota$ as the stationary queue length for a scaling parameter $\iota$ and $\mi(\iota) $ the corresponding value of $\mi$.
\begin{proposition}[Corollary.2 in \cite{van2021scaling}]
\label{prop:cor2}
  In stationary regime, as the scaling parameter $\iota \to \infty$,
  \begin{equation}
  \textstyle{
       c_f \iota^{\alpha-1}         {X_{i}^\iota } \rightarrow_{d.}  \chi_i \eqsp,}
\end{equation}
where $\chi_i=\CPE{E_i}{\sum\nolimits_{j=1}^{n_f} E_j/c_f \leq \beta}
$, $i \in [1, n_f]$,
and the $(E_j)_{j\le n}$ are independent unit mean exponential distributions.
\end{proposition}
As a consequence, using uniform integrability, we can estimate the following expected value (expected stationary queue lengths of fast, and slow nodes respectively) as follows:
\begin{equation}
    \begin{cases}
        \iota^{\alpha-1} \PE[X_i^{\iota}] \to \frac{\Gamma(c_f \beta)}{c_f}\eqsp, \quad \forall i \in [1,n_f],\\
       \iota^{\alpha-1}  \PE[X_i^{\iota}] \to \frac{1}{n-n_f} \left (\beta - n_f\frac{1}{c_f} \Gamma(c_f \beta) \right ), \forall i \in [n_f+1,n]\eqsp.
    \end{cases}
\end{equation}
Denoting  by  $P(k,x)= 1- \sum_{i=0}^{k-1} e^{-x} \frac{x^i}{i!}$, we have:
 $$\Gamma(c) =\frac{\PP(\sum\nolimits_{j=1}^{n_f+2} E_j\leq c)}{\PP(\sum\nolimits_{j=1}^{n_f+1} E_j \leq c)}\\
  =\frac{P(n_f+2, c)}{P(n_f+1, c)} \eqsp.
 $$
We now turn to bound the key quantity $\mi(\iota)$ for large $\iota$.
\begin{proposition}\label{prop:convSR}
Using the same assumptions as those of \Cref{prop:cor2} we get that :
\begin{equation}
   \begin{cases}
\limsup_{\iota \to \infty}   \iota^{\alpha-1} \mu_f \mi(\iota) &\le {\lambda} \frac{\Gamma(c_f \beta)}{c_f}\eqsp,  \quad \forall i \in [1,n_f] \eqsp,\\
  \limsup_{\iota \to \infty}   \iota^{\alpha-1} \mu_s \mi(\iota) &\le {\lambda} \frac{\beta - n_f \Gamma(c_f \beta)/c_f}{n-n_f}\eqsp,  \text{otherwise},
   \end{cases}
\end{equation}
where $\lambda= \sum_{i=1}^n \mu_i$.
\end{proposition}
We expect these bounds to be sharp for large $\iota$.

\paragraph{Numerical example} Under the previous assumptions, we have $\lambda=n_f \mu_f + (n-n_f) \mu_s$. We will further assume $n_f=\frac{n}{2}$, and $p_i=\frac{1}{n}$. Under these conditions, we have that $\Gamma(c_f \beta)$ is close to $1$. We can give a closed form approximations of the bounds of the expected delays:
\begin{equation}
    \begin{cases}
          \mi(\iota) \le \frac{n(\mu_f+\mu_s)}{2\mu_f(\mu_f/\mu_s -1)}\eqsp, \quad \forall i \in [1, n_f]\eqsp,\\
      \mi(\iota) \le \big(\frac{2\cardgrad}{n} - \frac{1}{\mu_f/\mu_s -1}\big)\frac{n(\mu_f+\mu_s)}{2\mu_s}\eqsp,  \quad \forall i \in [n_f +1, n].
    \end{cases}
\end{equation}
All delays bounds estimations have a closed form in the 2-cluster saturated regime: they only depend on the number of tasks in the network $\cardgrad$, on the number of nodes $n$, and on the intensity of nodes $\mu_f, \mu_s$. More details on the derivations, and on the following experiment are available in \Cref{sec:app:2clusters}. We consider a numerical simulation with $n=10$ clients, split in two clusters of same size: fast nodes with rate $\mu_f = 1.2$, and slow nodes with rates $\mu_s = 1$. We saturate the network with $\cardgrad = 1000$ tasks, and we simulate up to $\nsteps = 10^6$ server steps, and plot the distribution of the delays (in number of server steps).
\begin{figure}
    \centering
\includegraphics[scale=.5]{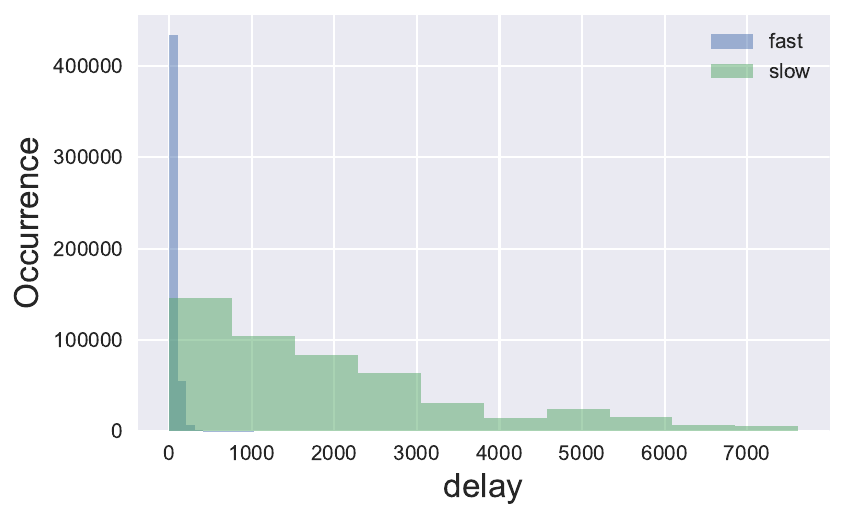}
    \caption{Histogram of fast and slow delays (in number of server steps) for a uniform sampling scheme.}
\label{fig:simu_delays2clusters}
\end{figure}
Our numerical experiment in \Cref{fig:simu_delays2clusters} gives average delays ($50$ and $1950$ for fast and slow nodes, respectively) and queue lengths that correspond to the theoretical expected values. It is also important to point out that the average delays are way smaller than the maximum delay experienced in the $\nsteps = 10^6$ steps. This further highlights the necessity to switch from analysis that depend on the $\tau_{max}$ quantity, to our analysis that only depends on the expected delays.

\section{Deep learning experiments}
\label{sec:experiments}
We evaluate FL algos performance on a classic image classification task: CIFAR-10 \citep{Krizhevsky2009}. We consider a non-\iid\ split of the dataset: each client takes seven classes (out of the ten possible) without replacement. This process introduces heterogeneity among the clients.

We compare different asynchronous methods in terms of CS steps. In all experiments, we track the performance of each algorithm by evaluating the server model against an unseen validation dataset.

We decide to focus on nodes with different exponential service rates as in \citet{nguyen2022federated}. We build \texttt{AsyncSGD} and \Algo\ codes from scratch. After simulating $n$ clients, we randomly group them into fast or slow nodes. We assume that the clients have different computational speeds, and refer the readers to \Cref{sec:app:simuruntime} for further details. We have assumed that half of the clients are slow. We compare the classic asynchronous methods \texttt{FedBuff} \citep{nguyen2022federated}, and \texttt{AsyncSGD} \citep{koloskova2022sharper}. Details about concurrent works implementation can be found in \Cref{sec:app:concu_works}.

We use the standard data augmentations and normalizations for all methods. All methods are implemented in Pytorch, and experiments are performed on an NVIDIA Tesla-P100 GPU. Standard multiclass cross entropy loss is used for all experiments. All models are fine-tuned with $n=100$ clients, and a batch of size $128$. We have finetuned the learning rate for each method. For \texttt{FedBuff} we tried several values for the buffer size, but finally found that the default one $Z=10$ gives the best performances.

In \Cref{fig:deep_cifar}, we compare the performance of a Resnet20 \citep{he2016deep} with the CIFAR-10 dataset, which consists of 50000 training images and 10000 test images (in 10 classes). The total number of CS steps is set to $200$. 
\begin{figure}
    \centering
    \includegraphics[scale=.5]{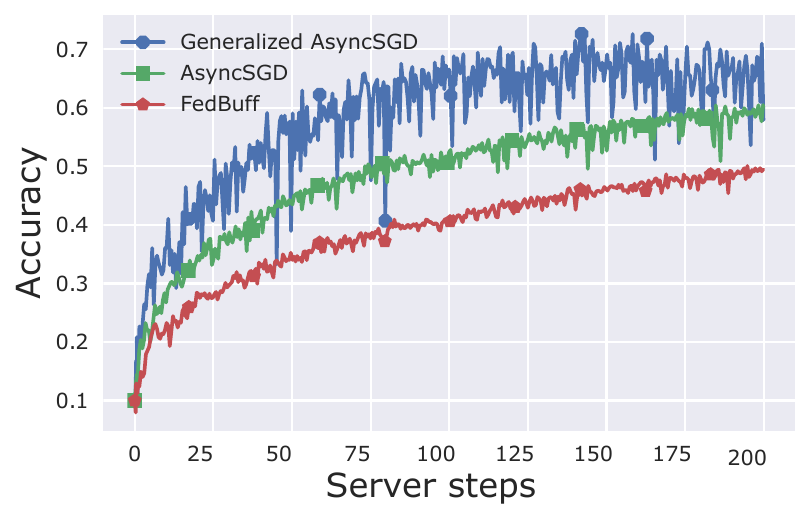}
    \caption{Accuracy on validation dataset on central server, for CIFAR-10 classification task.}
    \label{fig:deep_cifar}
\end{figure}
Despite the heterogeneity between client datasets, we can achieve good performance on image classification. \texttt{FedBuff} has to fill up its buffer before performing an update, slowing down the training process. \texttt{AsyncSGD} provides acceptable performance, but we can go further by sampling fast nodes slightly less than the uniform (as suggested in \Cref{sec:optim_bounds}), and this leads to much better accuracy.

We have additionally tested \Algo\ on the TinyImageNet classification task \citep{le2015tiny}, with a ResNet18. We compare \Algo\ with the classic synchronous approach FedAvg \citep{mcmahan2017communication} and two newer asynchronous methods FedBuff \citep{nguyen2022federated} and FAVANO \citep{leconte2023favas}. FAVANO (and the NN quantized QuAFL \cite{zakerinia2022quafl}) follows a completely different method than we do. There are no queues: Clients are triggered at the CS and either withhold their results or are interrupted by the CS before the work is completed. In FAVANO, the clients can have a high latency. The update rate of the CS is limited by (slow) clients: The minimum time between two CS updates should be at least as long as the minimum time needed to process a gradient update. TinyImageNet has $200$ classes and each class has $500$ (RGB) training images, $50$ validation images and $50$ test images. To train ResNet18, we follow the usual practices for training NNs: we resize the
input images to $64 \times 64$ and then randomly flip them horizontally during training. During testing, we center-crop them to the appropriate size. The learning rate is set to $0.001$ and the total simulated time is set to $1000$.
\begin{figure}
    \centering
    \includegraphics[scale=.5]{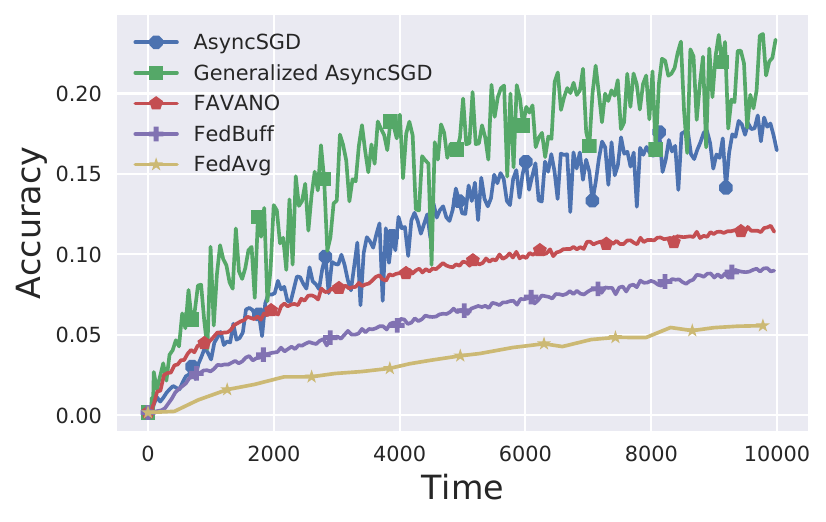}
    \caption{Test accuracy on TinyImageNet dataset with $n = 100$ total nodes.}
    \label{fig:tiny}
\end{figure}
\Cref{fig:tiny} illustrates the performance of \Algo\ in this experimental setup. While the partitioning of the training dataset follows an IID strategy,
TinyImageNet provides enough diversity to challenge federated learning algorithms. FedBuff is efficient when the number of stragglers is small. However, FedBuff is sensitive to the fraction of slow clients and may get stuck if the majority of clients in the buffer are more frequently the fast clients: this introduces a bias and few information from slow clients will be taken into account at the CS. FAVANO works better than FedBuff, but the CS updates should not be too small in order to allow slow clients to compute at least one local gradient step. However with AsyncSGD, no constraints are set on the time between two consecutive CS steps: it evolves freely based on the queuing processes. \Algo\ presents the same advantage, and in addition, samples clients with an optimal scheme. This leads to better performance, even on the challenging TinyImageNet benchmark.:

\section{Conclusion}
In this study, we analyze the convergence of an Asynchronous Federated Learning mechanism in a heterogeneous environment. Through a detailed queuing dynamics analysis, we demonstrate significantly improved convergence rates for our algorithm $\Algo$, eliminating dependence on the maximum delay $\tau_{\text{max}}$ seen in previous works. Our algorithm enables non-uniform node sampling, enhancing flexibility. Empirical evaluations reveal $\Algo$ superior efficiency over both synchronous and asynchronous state-of-the-art methods in standard CNN training benchmarks for image classification tasks.

\clearpage
\newpage
\bibliographystyle{apalike}
\bibliography{biblio}

\clearpage
\newpage
\appendix
\onecolumn
\section{Environmental footprint}
\label{sec:app_complexity}
In the current context, we estimated the carbon footprint of our experiments to be about 1 kg CO2e (calculated using green-algorithms.org v2.1 \cite{lannelongue2021green}). \Algo's time and memory complexity is on par with concurrent methods.

\section{Notations and definitions}
In the proof section below we will refer on the following notations. For $0 \leq k \leq \nsteps$ consider the filtration $\F_k$ defined as $\F_k = \sigma(\{w_\ell, \ell \leq k, K_m, m < k\}) $.
We define the \emph{virtual iterates} $\mu_k$ as follows:
\begin{equation}
\label{eq:virtual_iters_def}
\begin{cases}
    \mu_0 = \param_0,\\
    \mu_1 = \mu_0 -\eta \sum_{i \in \Set_0} \frac{1}{np_i}\Tilde{g}_{i}(\param_{0}),\\
    \mu_{{k+1}} = \mu_{k} - \frac{\eta}{np_{K_k}} \Tilde{g}_{K_k}(\param_{k})\eqsp, \quad k \geq 1\eqsp.
\end{cases}
\end{equation}

\section{Proofs of \Cref{sec:optim_bounds}}
\label{sec:app:proofs}
We split the proof of \Cref{theorem:cvgquant} into several steps. First we bound the quantity of interest $\sum\nolimits_{k=0}^{\nsteps} \PE[\norm{\nabla f(w_k)}^2]$ in terms of norms of difference between exact and virtual iterations $\norm{\mu_k - w_k}^2$ defined in \eqref{eq:virtual_iters_def}. More precisely, the following statement holds:

\begin{lemma}
\label{lem:technical-1}
Assume \Cref{assum:uniflowerbound} to \Cref{assum:graddissim} and let the learning rate $\eta$ satisfy $\eta \leq \frac{n^2}{8L\sum\nolimits_{i=1}^n\frac{1}{p_i}}$. Then for the iterates $(w_k)_{k \geq 0}$ of \Algo\ it holds that
\begin{equation}
\label{eq:virtual_iters_bound}
\frac{\eta}{4(\nsteps + 1)} \sum\nolimits_{k=0}^{\nsteps} \PE[\norm{\nabla f(w_k)}^2] \leq \frac{f(\mu_0) - \PE[f(\mu_{\nsteps+1})]}{\nsteps+1} + \frac{\eta L^2}{2} \frac{1}{\nsteps+1} \sum\nolimits_{k=0}^{\nsteps} \PE[\Vert \mu_k - w_k \Vert^2] + {\eta^2L}\sum\nolimits_{i=1}^n \frac{2{G}^2+\sigma^2}{n^2p_i}.
\end{equation}
\end{lemma}
\begin{proof}
Using the smoothness assumption \Cref{assum:smoothgradients} and the definition of $\mu_{k+1}$ from \eqref{eq:virtual_iters_def}, we obtain the following descent inequality:
\begin{align}
\CPE{f\left(\mu_{k+1}\right)}{\F_k} - f\left(\mu_{k}\right) & \leq - \eta \CPE{\langle \nabla f (\mu_k), \frac{1}{np_{K_k}} \widetilde{g}_{K_{k}}(w_k) \rangle}{\F_k} + \frac{\eta^2L}{2} \CPE{\Vert \frac{1}{np_{K_k}} \widetilde{g}_{K_{k}}(w_k) \Vert^2}{\F_k}\eqsp.
\end{align}
With the unbiasedness property of stochastic gradients and \Cref{assum:boundedvariance}, we get
\begin{align}
    \CPE{f\left(\mu_{k+1}\right)}{\F_k} - f\left(\mu_{k}\right) & \leq - \eta \CPE{\langle \nabla f (\mu_k), \frac{1}{np_{K_k}}\nabla f_{K_{k}}(w_k) \rangle}{\F_k} + \eta^2L \CPE{\Vert \frac{1}{np_{K_k}} \nabla f_{K_{k}}(w_k) \Vert^2}{\F_k} + \eta^2L\sigma^2\sum\nolimits_{i=1}^n\frac{1}{n^2p_i}\\
    & = - \eta \langle \nabla f (\mu_k), \nabla f(w_k) \rangle + \eta^2L \CPE{\Vert\frac{1}{np_{K_k}} \nabla f_{K_{k}}(w_k) \Vert^2}{\F_k}+ \eta^2L\sigma^2 \sum\nolimits_{i=1}^n\frac{1}{n^2p_i}\eqsp.
\end{align}
In the last equality we used that $\CPE{\langle \nabla f (\mu_k), \frac{1}{np_{K_k}}\nabla f_{K_{k}}(w_k) \rangle}{\F_k} = \langle \nabla f (\mu_k), \nabla f (w_k) \rangle$. Now we introduce a notation
\[
\Delta_k = \CPE{f\left(\mu_{k+1}\right)}{\F_k} - f\left(\mu_{k}\right)\eqsp.
\]
Since $\langle a, b \rangle = 1/2 (\norm{a}^2 + \norm{b}^2 - \norm{a-b}^2)$ for any $a,b \in \rset^{d}$, we get that
\begin{align}
    \Delta_k & \leq - \frac{\eta}{2} (\Vert \nabla f (\mu_k) \Vert^2 + \Vert \nabla f(w_k) \Vert^2 -\Vert \nabla f(w_k)-\nabla f(\mu_k) \Vert^2) + \eta^2L \CPE{\Vert \frac{1}{np_{K_k}} \nabla f_{K_{k}}(w_k) \Vert^2}{\F_k}+ \eta^2L\sigma^2\sum\nolimits_{i=1}^n\frac{1}{n^2p_i}\\
    & \leq -\frac{\eta}{2} \Vert \nabla f (\param_k) \Vert^2 +\frac{\eta}{2} L^2 \Vert \mu_k - \param_k \Vert^2 +\eta^2L \CPE{\Vert\frac{1}{np_{K_k}}( \nabla f_{K_{k}}(w_k)-\nabla f(w_k)+ \nabla f(w_k)) \Vert^2}{\F_k} + \eta^2L\sigma^2\sum\nolimits_{i=1}^n\frac{1}{n^2p_i}\\
    & \leq -\frac{\eta}{2} \Vert \nabla f (\param_k) \Vert^2 +\frac{\eta}{2} L^2 \Vert \mu_k - \param_k \Vert^2 +2\eta^2L(\CPE{\Vert\frac{1}{np_{K_k}}( \nabla f_{K_{k}}(w_k)-\nabla f(w_k)) \Vert^2}{\F_k} + \Vert \nabla f(w_k) \Vert^2\sum\nolimits_{i=1}^n\frac{1}{n^2p_i})\\
    &\qquad +\eta^2L\sigma^2\sum\nolimits_{i=1}^n\frac{1}{n^2p_i}.
\end{align}
Applying the bounded gradient dissimilarity assumption \Cref{assum:graddissim}, we get
\begin{align}
    \Delta_k & \leq -\frac{\eta}{2} \Vert \nabla f (\param_k) \Vert^2 +\frac{\eta}{2} L^2 \Vert \mu_k - \param_k \Vert^2 +{2\eta^2L}(\sum\nolimits_{i=1}^n \frac{{G}^2}{n^2p_i} + \Vert \nabla f(w_k) \Vert^2\sum\nolimits_{i=1}^n\frac{1}{n^2p_i})+ \eta^2L\sigma^2\sum\nolimits_{i=1}^n\frac{1}{n^2p_i} \\
    & \leq (-\frac{\eta}{2} + 2\eta^2L\sum\nolimits_{i=1}^n\frac{1}{n^2p_i}) \Vert \nabla f (\param_k) \Vert^2 +\frac{\eta}{2} L^2 \Vert \mu_k - \param_k \Vert^2 + {2\eta^2L}\sum\nolimits_{i=1}^n \frac{{G}^2}{n^2p_i}+ \eta^2L\sigma^2\sum\nolimits_{i=1}^n\frac{1}{n^2p_i}.
\end{align}
As a consequence by taking $\eta \leq \frac{n^2}{8L\sum\nolimits_{i=1}^n\frac{1}{p_i}}$ and substituting for $\Delta_k$, we get
\begin{align}
    \frac{\eta}{4} \Vert \nabla f(w_k) \Vert ^2 \leq f(\mu_k) - \CPE{f(\mu_{k+1})}{\mathcal{F}_k} + \frac{\eta L^2}{2} \Vert \mu_k - w_k \Vert^2 + {2\eta^2L}\sum\nolimits_{i=1}^n \frac{{G}^2}{n^2p_i}+ \eta^2L\sigma^2\sum\nolimits_{i=1}^n\frac{1}{n^2p_i}.
\end{align}
Now taking sum for $k \in \{0,\ldots, \nsteps\}$, we get
\begin{align}
\frac{\eta}{4(\nsteps+1)} \sum\nolimits_{k=0}^{\nsteps} \PE[\Vert \nabla f(w_k) \Vert ^2] \leq \frac{f(\mu_0) - \PE[f(\mu_{K+1})]}{\nsteps+1} + \frac{\eta L^2}{2} \frac{1}{\nsteps+1} \sum\nolimits_{k=0}^{\nsteps} \PE[\Vert \mu_k - w_k \Vert^2] + {\eta^2L}\sum\nolimits_{i=1}^n \frac{2{G}^2+\sigma^2}{n^2p_i}.
\end{align}

\end{proof}

In order to apply the result of \Cref{lem:technical-1}, one needs to provide an upper bound on the correction term $\PE[\norm{\mu_k - \param_k}^2]$.
As explained in \Cref{sec:optim_bounds}, the virtual iterates deviation from the true $\{ \param_k \}_{k>0}$ is made of all the gradients computed (on potentially outdated $\param$'s) and not applied yet. We can introduce the sets $\{ \mathcal{I}_k \}_{k>0}$, as the sets of time and client indexes whose gradients are still on fly at time $k$. With $\Set_0$ being the set of initial active workers from \Algo, there are defined by the recursion:
\begin{equation}
\mathcal{I}_1 = \{ (i,0) | i \in \Set_0, i \neq J_0 \}
\end{equation}
\begin{equation}
\mathcal{I}_{k+1}=
\begin{cases}
     \mathcal{I}_k \quad \text{if } I_k=k,\\
     \mathcal{I}_k \setminus (J_k, I_k)\cup (K_k, k) \quad \text{otherwise}.
\end{cases}
\end{equation}
As $\norm{\mu_k - \param_k}$ represents the norm of gradients, it is easier to introduce the sets $\{ \mathcal{G}_k \}_{k>0}$, as the set of gradients \emph{scaled} with their respective weight $\frac{-1}{np_i}$ for each client $i$, that correspond to the indexes in $\{ \mathcal{I}_k \}_{k>0}$:
\begin{equation}
    \mathcal{G}_k  = \{ -\frac{1}{np_i}\Tilde{g}_i(\param_j) | (i,j) \in \mathcal{I}_k \}.
\end{equation}
In the following lines, we will show that the sets $\{ \mathcal{G}_k \}_{k>0}$ (and as a consequence the sets $\{ \mathcal{I}_k \}_{k>0}$) have a constant cardinal: the number of running tasks in \Algo\ is fixed during the whole optimization process, and only depends on the initialization.
\begin{remark}
    The number of running tasks is constant, but the number of active nodes is not! If there is a very slow client $i$, the number of active clients can be reduced to $1$: all tasks are currently processed in the queue of client $i$.
\end{remark}

\begin{lemma}
\label{lem:technical-lemmma-2}
For the sequence $(w_k)_{k \geq 0}$ of updates produced by \Algo\ and for the sequence of virtual updates $(\mu_k)_{k \geq 0}$ defined in \eqref{eq:virtual_iters_def}, it holds that
\begin{align}
\mu_1 - \param_1   &= -\eta \sum\nolimits_{i \in \Set_0} \indin{i \neq J_0}\frac{1}{np_i}\Tilde{g}_i(\param_0)\eqsp, \\
\mu_{k+1} - \param_{k+1} & =  -\eta \sum\nolimits_{i \in \Set_0} \indin{i \neq J_0} \frac{1}{np_{i}}\Tilde{g}_i(\param_0) +\eta  \sum\nolimits_{r=1}^k (\frac{1}{np_{J_r}}\Tilde{g}_{J_r}(w_{I_r}) - \frac{1}{np_{K_r}}\Tilde{g}_{K_{r}}(w_r))\eqsp, \quad k \geq 1\eqsp.
\end{align}
\end{lemma}
\begin{proof}
The proof follows from the definition of recurrence \eqref{eq:virtual_iters_def}. Indeed, first we can note that
\begin{align}
     \mu_1 - \param_1   &=  (w_0 - \eta \sum\nolimits_{i \in \Set_0} \frac{1}{np_i}\Tilde{g}_i(\param_0) ) - (w_0-\eta \frac{1}{np_{J_0}}\Tilde{g}_{J_{0}}(\param_{I_0})) \\
    &=  (w_0 - \eta \sum\nolimits_{i \in \Set_0} \frac{1}{np_i}\Tilde{g}_i(\param_0)) - (w_0-\eta \frac{1}{np_{J_0}}\Tilde{g}_{J_{0}}(\param_{0})) \\
    &=  -\eta \sum\nolimits_{i \in \Set_0} \indin{i \neq J_0}\frac{1}{np_i}\Tilde{g}_i(\param_0) .
\end{align}
Now for a general iteration number $k$ we have:
\begin{align}
    \mu_{k+1} - \param_{k+1}   &= (\mu_k - \eta \frac{1}{np_{K_k}}\Tilde{g}_{K_{k}}(\param_k)) - (w_k-\eta \frac{1}{np_{J_k}}\Tilde{g}_{J_k}(\param_{I_k})) \\
    &= (\mu_k - \param_k) + \eta(\frac{1}{np_{J_k}}\Tilde{g}_{J_k}(\param_{I_k})- \frac{1}{np_{K_k}}\Tilde{g}_{K_{k}}(\param_k)) \\
    &= (\mu_1 - \param_1) + \sum\nolimits_{r=1}^k \eta(\frac{1}{np_{J_r}}\Tilde{g}_{J_r}(\param_{I_r})- \frac{1}{np_{K_r}}\Tilde{g}_{K_{r}}(\param_r)) \\
    & =  -\eta \sum\nolimits_{i \in \Set_0} \indin{i \neq J_0} \frac{1}{np_{i}}\Tilde{g}_i(\param_0) +\eta  \sum\nolimits_{r=1}^k (\frac{1}{np_{J_r}}\Tilde{g}_{J_r}(w_{I_r}) - \frac{1}{np_{K_r}}\Tilde{g}_{K_{r}}(w_r)).
\end{align}
\end{proof}

\begin{lemma}
\label{lem:cardinality}
The sets $\{\mathcal{G}_k\}_{k \geq 0}$ have constant cardinal and compile all the gradients in computation at step $k>0$:
    \begin{equation}
        \begin{cases}
            (i) \quad |\mathcal{G}_k | = |\mathcal{G}_1 | = |\Set_0 | - 1,\\
            (ii) \quad  \mu_k - \param_k  =  \eta \sum_{g \in \mathcal{G}_k} g .
        \end{cases}
    \end{equation}
\begin{proof}
\textbf{Step (i)}: We are going to proof the result by induction. Assume $|\mathcal{G}_k | = |\Set_0 | - 1$ for some $k$. If $I_k=k$ we can immediately conclude that $|\mathcal{G}_{k+1} | = |\mathcal{G}_k |$.  Otherwise, $I_k < k$, hence there exists some $i \in [n]$ such that $- \frac{1}{np_{i}}\Tilde{g}_i(w_{I_k}) \in \mathcal{G}_k$. In particular, client $J_k$ is the client that finishes computation at step $k$, thus $- \frac{1}{np_{J_k}}\Tilde{g}_{J_k}(w_{I_k}) \in \mathcal{G}_k$. As a consequence, $|\mathcal{G}_k \setminus \{-\frac{1}{np_{J_k}}\Tilde{g}_{J_k}(w_{I_k})\}| = |\Set_0| -2$. Furthermore, by definition, all gradients in $\mathcal{G}_k$ are taken on models older than $k$. And by taking $I_k<k$, we obtain $-\frac{1}{np_{K_k}}\Tilde{g}_{K_k}(w_{k}) \notin \Set_k \setminus \{-\frac{1}{np_{J_k}}\Tilde{g}_{J_k}(w_{I_k})\}$. It concludes $|\mathcal{G}_{k+1} | = |\mathcal{G}_k |$.

\textbf{Step (ii)}: We also prove it by induction. It is valid for $k=1$. Now assume $\quad  \mu_k - \param_k = \sum_{g \in \mathcal{G}_k} g$, for some $k>1$.
\begin{align}
     \mu_{k+1} - \param_{k+1}  &=  (\mu_k - \param_k) + \frac{1}{np_{J_k}}\Tilde{g}_{J_k}(w_{I_k}) - \frac{1}{np_{K_k}}\Tilde{g}_{K_k}(w_k) \\
    &=  (\mu_k - \param_k) + (\indiacc{0}(I_k)\frac{1}{np_{J_k}}\Tilde{g}_{J_k}(w_{0})+ \dots + \indiacc{k}(I_k)\frac{1}{np_{J_k}}\Tilde{g}_{J_k}(w_{k})) - \frac{1}{np_{K_k}}\Tilde{g}_{K_k}(w_k) \\
    &=  (\mu_k - \param_k) + \indiacc{0}(I_k)\frac{1}{np_{J_k}}\Tilde{g}_{J_k}(w_{0})+ \dots + \indiacc{k-1}(I_k)\frac{1}{np_{J_k}}\Tilde{g}_{J_k}(w_{k-1})\\
    & \quad + (\indiacc{k}(I_k) \frac{1}{np_{J_k}}\Tilde{g}_{J_k}(w_{k}) - \frac{1}{np_{K_k}}\Tilde{g}_{K_k}(w_k)).
\end{align}
By induction, we have:
\begin{align}
     \mu_{k+1} - \param_{k+1}  &=  \sum\nolimits_{g \in \mathcal{G}_k} g + \indiacc{0}(I_k) \frac{1}{np_{J_k}}\Tilde{g}_{J_k}(w_{0})+ \dots + \indiacc{k-1}(I_k)\frac{1}{np_{J_k}}\Tilde{g}_{J_k}(w_{k-1}) + (\indiacc{k}(I_k)\frac{1}{np_{J_k}}\Tilde{g}_{J_k}(w_{k}) - \frac{1}{np_{K_k}}\Tilde{g}_{K_k}(w_k)).
\end{align}
If $I_k=k$ we have $J_k=K_k$: some client contributes instantaneously. This results in:
\begin{align}
     \mu_{k+1} - \param_{k+1}  &=  \sum\nolimits_{g \in \mathcal{G}_k} g + \underbrace{\indiacc{0}(I_k)\frac{1}{np_{J_k}}\Tilde{g}_{J_k}(w_{0})+ \dots + \indiacc{k-1}(I_k)\frac{1}{np_{J_k}}\Tilde{g}_{J_k}(w_{k-1})}_{=0} + \underbrace{(\indiacc{k}(I_k) \Tilde{g}_{J_k}(w_{k}) - \frac{1}{np_{K_k}}\Tilde{g}_{K_k}(w_k))}_{=0}\\
     &= \sum\nolimits_{g \in \mathcal{G}_{k+1}} g.
\end{align}
Same idea as in Step (i) allows us to conclude $\mu_{k+1} - \param_{k+1} =  \sum_{g \in \mathcal{G}_{k+1}} g $ when $I_k < k$.
\end{proof}
\end{lemma}

\begin{lemma}
\label{lem:technical-lemma-3}
For the sequence $(w_k)_{k \geq 0}$ of updates produced by \Algo\ and for the sequence of virtual updates $(\mu_k)_{k \geq 0}$ defined in \eqref{eq:virtual_iters_def}, it holds that
\begin{align}
\frac{1}{\nsteps+1} \sum\nolimits_{k=0}^{\nsteps} \PE[\Vert \mu_k - \param_k \Vert^2] & \leq 2 \eta^2 \cardgrad \sum\nolimits_{i=1}^n  \frac{\m[i][0][\nsteps]}{(n^2p_i^2)(\nsteps+1)} (2G^2+\sigma^2)\\
& \quad + 4 \eta^2 \cardgrad  \frac{\PE[\Vert \nabla f(w_0) \Vert^2] \weightm[0][\nsteps]}{(n^2p_i^2)(\nsteps+1)} \\
& \quad + 4 \eta^2 \cardgrad \sum\nolimits_{k=1}^{\nsteps} \frac{\weightm}{(n^2p_i^2)(\nsteps+1)}\PE[\Vert \nabla f(w_k) \Vert^2]\\
& \quad + 2 \eta^2 \cardgrad \sum\nolimits_{i=1}^n \frac{\sum\nolimits_{k=1}^K \m}{(n^2p_i^2)(\nsteps+1)}(2G^2+\sigma^2).
\end{align}
\end{lemma}
\begin{proof}
Using the statement of \Cref{lem:cardinality}, we bound the expected value of the correction term as follows:
\begin{align}
\label{eq:first_bound_corr}
\PE[\Vert \mu_k - \param_k \Vert^2] =\eta^2 \PE[\Vert \sum\nolimits_{g \in \mathcal{G}_k}  g \Vert^2] \leq \eta^2 \PE[|\mathcal{G}_k| \sum\nolimits_{g \in \mathcal{G}_k} \Vert g \Vert^2] \leq \eta^2 \PE[\cardgrad \sum\nolimits_{g \in \mathcal{G}_k} \Vert g \Vert^2].
\end{align}
As the cardinal of sets $|\mathcal{G}_k|:=\cardgrad$ is constant among iterations, and in particular it is independent from $g$'s, we can simplify the form above.

We also introduce the set $U_k$:
\begin{equation}
U_k = \{ i \in \{1, \ldots, n\} | X_{i,k}>0 \}.
\end{equation}
Hence, from \eqref{eq:first_bound_corr} and \Cref{assum:boundedvariance}, we get
\begin{align}
    \PE[\Vert \mu_k - \param_k \Vert^2] &\leq \eta^2 \cardgrad \PE[\sum\nolimits_{(i,j) \in \mathcal{I}_k \cup \{U_k \times \{0\}\}} \frac{1}{n^2p_i^2}2\Vert \nabla f_i(w_j) \Vert^2+ 2\sigma^2]\\
    & \leq \eta^2 \cardgrad \PE[\sum\nolimits_{i=1}^n \frac{1}{n^2p_i^2} \underbrace{\indi{U_k \cap \Set_0}(i)(4G^2+4 \Vert \nabla f(w_0) \Vert^2+2\sigma^2)}_{\text{gradients on initial model}} + \sum\nolimits_{j=1}^k \indi{\mathcal{I}_k}\bigl((i,j)\bigr) (4G^2 + 4 \Vert \nabla f(w_j) \Vert^2+2\sigma^2)]
\end{align}
Now we average over $\nsteps$ iterations:
\begin{align}
    \frac{1}{\nsteps+1} \sum\nolimits_{k=0}^{\nsteps} \PE[\Vert \mu_k - \param_k \Vert^2] & \leq 2 \eta^2 \cardgrad \sum\nolimits_{i=1}^n (\sum\nolimits_{k=0}^{\nsteps} \frac{\PP( i \in U_k \cap \Set_0)}{\nsteps+1}) (\frac{2G^2+\sigma^2}{n^2p_i^2})\\
    & \quad + 4 \eta^2 \cardgrad  \PE[\Vert \nabla f(w_0) \Vert^2] \sum\nolimits_{i=1}^n \frac{1}{n^2p_i^2}(\sum\nolimits_{k=0}^{\nsteps} \frac{\PP(i \in U_k \cap \Set_0)}{\nsteps+1}) \\
    & \quad + \frac{4}{\nsteps+1} \eta^2 \cardgrad \sum\nolimits_{i=1}^n \PE[ \sum\nolimits_{k=1}^{\nsteps} \frac{1}{n^2p_i^2} \sum_{j=1}^k(\indi{\mathcal{I}_k}\bigl((i,j)\bigr)\Vert \nabla f(w_j) \Vert^2)]\\
    & \quad + \frac{4G^2+2\sigma^2}{\nsteps+1} \eta^2 \cardgrad \sum\nolimits_{i=1}^n \PE[ \sum\nolimits_{k=1}^{\nsteps} \frac{1}{n^2p_i^2} \sum\nolimits_{j=1}^k\indi{\mathcal{I}_k}((i,j)) ].
\end{align}
We rearrange the last 2 terms:
\begin{align}
    \frac{1}{\nsteps+1} \sum\nolimits_{k=0}^{\nsteps} \PE[\Vert \mu_k - \param_k \Vert^2] & \leq 2 \eta^2 \cardgrad \sum\nolimits_{i=1}^n  \frac{\sum\nolimits_{k=1}^{\nsteps}\PP(i \in U_k \cap \Set_0)}{\nsteps+1} (\frac{2G^2+\sigma^2}{n^2p_i^2})\\
    & \quad + 4 \eta^2 \cardgrad \PE[\Vert \nabla f(w_0) \Vert^2] \sum\nolimits_{i=1}^n \frac{1}{n^2p_i^2} \frac{\sum\nolimits_{k=1}^{\nsteps}\PP(i \in u_k \cap \Set_0)}{\nsteps+1} \\
    & \quad + 4 \eta^2 \cardgrad \sum\nolimits_{k=1}^{\nsteps} \PE[\frac{\sum\nolimits_{i=1}^n \frac{1}{n^2p_i^2} \sum\nolimits_{r=k}^{\nsteps}\indi{\mathcal{I}_r}((i,k))}{\nsteps+1}\Vert \nabla f(w_k) \Vert^2]\\
    & \quad + 2 \eta^2 \cardgrad \sum\nolimits\nolimits_{k=1}^{\nsteps} \PE[\frac{\sum\nolimits_{i=1}^n \frac{1}{n^2p_i^2} \sum\nolimits_{r=k}^{\nsteps}\indi{\mathcal{I}_r}((i,k))}{\nsteps+1}(2G^2+\sigma^2)].
\end{align}

We can simplify the bounds with the following identity:
\begin{equation}
\sum\nolimits_{r=k}^{\nsteps} \indi{\mathcal{I}_r}((i,k)) = \indiacc{i}(K_{k+1})\sum\nolimits_{r=k}^{\nsteps} \indi{(\sum\nolimits_{l=k}^r\indi{J_l=i})<X_{i,k}} = \M[i][k][\nsteps], \text{for } k>0.
\end{equation}
And with a slight abuse of notation, we take: $\M[i][0][\nsteps] = \sum\nolimits_{k=1}^\nsteps\indi{U_k \cap \Set_0}(i).$

Combining the above bounds, we obtain
\begin{align}
    \frac{1}{\nsteps+1} \sum\nolimits_{k=0}^{\nsteps} \PE[\Vert \mu_k - \param_k \Vert^2] & \leq 2 \eta^2 \cardgrad \sum\nolimits_{i=1}^n  \frac{\PE[\M[i][0][\nsteps]]}{\nsteps+1} (\frac{2G^2+\sigma^2}{n^2p_i^2})\\
    & \quad + 4 \eta^2 \cardgrad  \PE[\Vert \nabla f(w_0) \Vert^2] \sum\nolimits_{i=1}^n \frac{1}{n^2p_i^2} \frac{\PE[\M[i][0][\nsteps]]}{\nsteps+1} \\
    & \quad + 4 \eta^2 \cardgrad \sum\nolimits_{k=1}^{\nsteps} \frac{\sum_{i=1}^n \frac{1}{n^2p_i^2}\PE[\M[i][k][\nsteps]]}{\nsteps+1}\PE[\Vert \nabla f(w_k) \Vert^2]\\
    & \quad + 2 \eta^2 \cardgrad \sum\nolimits_{i=1}^n \frac{\sum\nolimits_{k=1}^{\nsteps}\PE[\M[i][k][\nsteps]]}{\nsteps + 1}(\frac{2G^2+\sigma^2}{n^2p_i^2})\eqsp,
\end{align}
and the statement follows using the definition $\weightm$.
\end{proof}

\subsection{Proof of \Cref{theorem:cvgquant}}
We apply the bound \Cref{lem:technical-1} and use \Cref{lem:technical-lemma-3} to control the correction term $\frac{1}{\nsteps+1} \sum_{k=0}^{\nsteps} \PE[\Vert \mu_k - \param_k \Vert^2]$. Hence, we get
\begin{align}
    \frac{1}{4(\nsteps+1)} \sum\nolimits_{k=0}^{\nsteps} \PE[\Vert \nabla f(w_k) \Vert ^2] & \leq \frac{\PE[f(\mu_0) - f(\mu_{\nsteps+1})]}{\eta (\nsteps + 1)} + \frac{L^2}{2(\nsteps+1)} \sum\nolimits_{k=0}^{\nsteps} \PE[\Vert \mu_k - w_k \Vert^2] + {\eta L}\sum\nolimits_i^n \frac{2G^2+\sigma^2}{n^2p_i}\\
    & \leq \frac{\PE[f(\mu_0) - f(\mu_{\nsteps+1})]}{\eta (\nsteps+1)} + {\eta L}\sum\nolimits_i^n \frac{2G^2+\sigma^2}{n^2p_i}\\
    & \quad + L^2\eta^2 \cardgrad \left( \sum\nolimits_{i=1}^n  \frac{\sum\nolimits_{k=0}^{\nsteps} \m[i][k][\nsteps]}{n^2 p_i^2 (\nsteps+1)} (2G^2+\sigma^2) +  \frac{2\PE[\Vert \nabla f(w_0) \Vert^2] \weightm[0][\nsteps]}{\nsteps+1} \right) \\
    & \quad + \frac{L^2\eta^2 \cardgrad}{\nsteps+1} \sum\nolimits_{k=1}^{\nsteps} 2 \weightm[k][\nsteps] \PE[\Vert \nabla f(w_k) \Vert^2]\eqsp.
\end{align}
Hence we have:
\begin{align}
\frac{1}{\nsteps+1} \sum\nolimits_{k=0}^{\nsteps} (1/4 - 2\weightm[k][\nsteps] L^2\eta^2 \cardgrad) \PE[\Vert \nabla f(w_k) \Vert ^2] & \leq \frac{\PE[f(\mu_0) - f(\mu_{\nsteps+1})]}{\eta (\nsteps+1)} + {\eta L}\sum\nolimits_i^n \frac{2G^2+\sigma^2}{n^2p_i}\\
& \quad + L^2\eta^2 \cardgrad \sum\nolimits_{i=1}^n  \frac{\sum\nolimits_{k=0}^{\nsteps} \m[i][k][\nsteps]}{n^2 p_i^2 (\nsteps+1)} (2G^2+\sigma^2)\eqsp.
\end{align}
Now we impose the step size condition
\begin{equation}
\label{eq:max_step_size}
\eta \leq \sqrt{\frac{1}{16L^2 \cardgrad \max_{k \in \{1,\ldots,\nsteps\}}\weightm[k][\nsteps]}}\eqsp,
\end{equation}
which enable us to conclude that
\begin{align}
    \frac{1}{8(\nsteps+1)} \sum\nolimits_{k=0}^{\nsteps} \PE[\Vert \nabla f(w_k) \Vert ^2] & \leq \frac{\PE[f(\mu_0) - f(\mu_{\nsteps+1})]}{\eta (\nsteps+1)} + {\eta L}\sum\nolimits_i^n \frac{2G^2+\sigma^2}{n^2p_i}\\
    & \quad + L^2\eta^2 \cardgrad \sum\nolimits_{i=1}^n  \frac{\sum\nolimits_{k=0}^{\nsteps} \m[i][k][\nsteps]}{n^2p_i^2(\nsteps+1)} (2G^2+\sigma^2)\eqsp,
\end{align}
and the statement follows.

\subsection{Influence of the strong growth condition}
The assumption \Cref{assum:boundedvariance} can be generalized to the \emph{strong growth condition} \cite{vaswani2019fast}:
\[
\textstyle{
\mathbb{E}[\left\|\widetilde{g}_i(x)-\nabla f_i(x)\right\|^2] \leq \sigma^2 + \rho^2 \left\| \nabla f_{i}(x)\right\|^2
}\eqsp.
\]
All previous derivations are impacted by a factor $\rho^2$. But the proofs remain the same. In particular, the quantity $\Delta_k$ from \Cref{lem:technical-1} can be bounded as:
\begin{align}
    \Delta_k & \leq (-\frac{\eta}{2} + 2\eta^2L\sum\nolimits_{i=1}^n\frac{1+\rho^2}{n^2p_i}) \Vert \nabla f (\param_k) \Vert^2 +\frac{\eta}{2} L^2 \Vert \mu_k - \param_k \Vert^2 + {2\eta^2L}\sum\nolimits_{i=1}^n \frac{(1+\rho^2){G}^2}{n^2p_i}+ \eta^2L\sigma^2\sum\nolimits_{i=1}^n\frac{1}{n^2p_i}.
\end{align}
This slightly change the condition on the step size: $\eta \leq \frac{n^2}{8L\sum\nolimits_{i=1}^n\frac{1+\rho^2}{p_i}}$.
The rest of the proof is similarly impacted. We now impose the additional step size condition:
\begin{equation}
\eta \leq \sqrt{\frac{1}{(1+\rho^2)16L^2 \cardgrad \max_{k \in \{1,\ldots,\nsteps\}}\weightm[k][\nsteps]}}\eqsp,
\end{equation}
which enable us to conclude that
\begin{align}
    \frac{1}{8(\nsteps+1)} \sum\nolimits_{k=0}^{\nsteps} \PE[\Vert \nabla f(w_k) \Vert ^2] & \leq \frac{\PE[f(\mu_0) - f(\mu_{\nsteps+1})]}{\eta (\nsteps+1)} + {\eta L}\sum\nolimits_i^n \frac{2(1+\rho^2)G^2+\sigma^2}{n^2p_i}\\
    & \quad + L^2\eta^2 \cardgrad \sum\nolimits_{i=1}^n  \frac{\sum\nolimits_{k=0}^{\nsteps} \m[i][k][\nsteps]}{n^2p_i^2(\nsteps+1)} (2(1+\rho^2)G^2+\sigma^2)\eqsp.
\end{align}

\section{Proofs of Section \ref{sec:Jack}}
\label{sec:app_proof_queuing}

\subsection{Proof of Proposition \ref{prop:Jack}}

Given the assumptions, it is straightforward to check that the dynamics are those of the Markov process with the given generator.

Then the result follows from classical results in queuing theory: the Markov process corresponds to a Jackson quasi-reversible network with an explicit stationary distribution, see for instance Theorem 1.12 in \cite{serfozo1999}.

Before continuing, we need a fundamental property of closed Jackson network which is the arrival Theorem also called MUSTA in the literature:

\begin{theorem}[Arrival Theorem, Prop 4.35 in \cite{serfozo1999}]\label{theo:arr}
Suppose the system is stationary.
Upon arrival to a given node (i.e., just before waiting or being served in the queue), a task sees the network according to the distribution $\pi_{C-1}$, i.e.,
$$ \PP( X_{{\tau_{i,1}}^-} =x)=\pi_{C-1}(x).$$
\end{theorem}
For the link with Palm probabilities, we also refer to Example 3.3.4. in \cite{baccelli2002}.

Now, define $S_i$ the first sojourn time on node $i$, i.e.,
\begin{equation}
    S_i = \inf\{t \ge \tau_{i,1} | D_i(t) = X_i({\tau_{i,1}}^{-}) +1 \}
\end{equation}


Recall that,
$\PE^{C}$ corresponds to the stationary average for a system with $C$ tasks (in particular the process $X_t$ follows $\pi_C$ for all times $t \ge 0$).
We denote in turn by $\PE^{C-1}$ the Palm probability associated to the event of an arrival at a given node (say $i$) and corresponding informally to conditioning to $\tau_{i,1}= 0$. Using the Arrival Theorem,
the distribution at time $0$ at node $i$ corresponds in this case to $\pi_{C-1}$, the dynamics under the Palm probabilities being unchanged (see\cite{baccelli2002}).

\subsection{Proof of Proposition \ref{prop:mik}}

Assuming the system is stationary and using the definition of $\m$ we have that for any $k$,
\begin{align}
\m = \weightm[i][\nsteps] = \PE^{\cardgrad}\Big[ \Big(\sum\nolimits_{n} \mathds{1}(\tau_{i,1} \le T_n \le S_i) \Big)\wedge \nsteps \big].
\end{align}
Using the monotone convergence Theorem,
\begin{align}
\label{eq:delay}
\lim_{\nsteps \to \infty} \m = m_{i} =\PE^{\cardgrad}\Big[ \sum\nolimits_{n} \mathds{1}(\tau_{i,1} \le T_n \le S_i) \big].
\end{align}

We then use the arrival Theorem (\ref{theo:arr}) for closed Jackson network and using the Palm probability, we can write that
\small
$$
\PE^{\cardgrad}[\sum\nolimits_{n} \mathds{1}(\tau_{i,1} \le T_n \le S_i)] =  \PE^{{\cardgrad-1}}[\sum\nolimits_{n} \mathds{1}(0\le T_n \le S_i)].$$
\normalsize
Then by using the stochastic intensity formula (see Definition 1.8.10 and Example
1.8.3. in \cite{baccelli2002}):
\begin{align}
\PE^{{\cardgrad-1}}[\sum\nolimits_{n} \mathds{1}(0 \le T_n \le S_i)] &=
\PE^{{\cardgrad-1}}\Big[\int_0^{S_i} \sum\nolimits_{j=1}^{n}\mu_j \mathds{1}(X_j(s)>0) \,ds\Big].
\end{align}

\subsection{Computation of the constant $\Gamma$}

\begin{align}
    \Gamma(c) &=    {\PE[ X  \indi{X+Y \le c}] \over \PE[  \indi{X+Y \le c}]} ,
\end{align}
with $X$ an Exp(1) and  $Y$ an Erlang(F,1) independent from each other.
By integrating by parts:
\begin{align}
       \Gamma  &= \frac{ \int_0^\infty\int_0^{ c- y} x e^{-x}dx \mathds{1}( y\leq c) dP_Y(y)}{\PP(X+Y \le c)},\\
      &= \frac{ \int_0^\infty (-(c-y)e^{-c+y} + \int_x 1_{x \le c-y} \mathds{1}( y\leq c) dP_X(x) dP_Y(y)}{\PP(X+Y \le c)},\\
            &= \frac{ \int_0^\infty -(c-y)e^{-c+y}   \mathds{1}( y\leq c) {y^{F-1} \over (F-1)!}e^{-y}dy}{\PP( X+Y \le c)} +1,\\
         &= \frac{ e^{-c}( -c^{F+1}/F! + F c^{F+1}/(F+1)!) }{\PP(X+Y \le c)} +1,\\
     &= \frac{ -e^{-c}c^{F+1}/(F+1)! + 1- \sum_{k=1}^{F} e^{-c}c^{k}/(k)! }{1-\sum_{k=1}^{F} e^{-c}c^{k}/(k)!},\\
  &= \frac{{\PP(\sum_{i=1}^{F+2} E_i \le c)}}{{\PP(\sum_{i=1}^{F+1} E_i \le c)}},\\
     \end{align}

\subsection{Proof of Proposition \ref{prop:convSR}}

First note that:
$$ \mi(\iota) = \PE^{{\cardgrad-1}}\Big[\int_0^{S_i} \sum\nolimits_{j=1}^{n}\mu_j \mathds{1}(X_j^{\iota}(s)>0) \,ds\Big] \le \lambda \PE^{{\cardgrad-1}}[S_i].$$

Then it follows from the FIFO representation that
$$\PE^{{\cardgrad-1}}(S_i)={1 \over \mu_i} (\PE^{{\cardgrad-1}}[X_i^{\iota}]) +1)$$

which implies the claim.
\section{Upper bounds simulations}
\label{sec:app:bounds_simu}
\subsection{Illustration of $G(\mathbf{p},\eta)$ before optimization}
We decide to simulate $n=100$ nodes with $\cardgrad=10$ initial tasks. Nodes can only be fast (sampled with $p_i=p$) or slow (sampled with $p_i=\frac{2}{n}-p$), and are evenly distributed. We estimate the values of $\m$ through Monte-Carlo and compute the upper bounds given in \Cref{theorem:cvgquant}. All others constants are kept unitary to ease the computation.
\begin{figure}[h!]
    \centering
    \includegraphics[scale=.6]{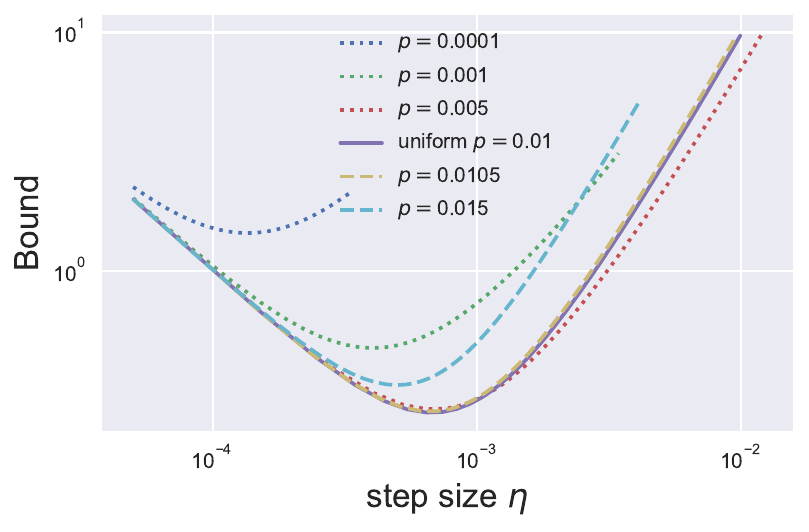}
    \caption{Variation of the non-convex upper-bound with respect to the step size $\eta$, for $K=10^4$ server steps and different values of sampling $p$. The maximum step size value is different for each case and equal to $\sqrt{\frac{1}{8L^2 \cardgrad max \weightm }}$.}
    \label{fig:bounds}
\end{figure}
In \Cref{fig:bounds}, we plot the value of the previously mentioned upper-bound with respect to the step size $\eta$, for several sampling probabilities of fast node $p$. When the step size considered is small, all sampling strategies are equivalent. Whereas for large value of $\eta$, sampling around the uniform one is a good strategy. Large value of $p$, close to the limit $\frac{2}{n}$, hinders the bound because it increases the delays for slow nodes by sampling fast nodes quite often. In \Cref{fig:proba_optimal} and \Cref{fig:relative_bounds}, we define grids of $50$ values of $p$ around the uniform one, and for each $p$ we compute the exact optimal step size by solving the cube roots. The optimal values of the sampling $p$ on the grid, and of the optimal step size are further used to compute the optimal bound and compare it against the one obtained with uniform sampling.
\subsection{Bounds w.r.t physical time}
\label{sec:app:physical_time}
We want here to focus on the relative improvements of the upper bounds when time rather than CS steps is considered as fixed. Indeed, when we determine complexity in terms of number of communications, we don't take into account the time intervals between two successive arrivals at the central server. In particular, the results from \Cref{sec:optim_bounds} propose to sample more frequently slow nodes. But this results into an increased waiting time between two consecutive server steps. As a consequence, in this section, we choose a fixed unit of time $U=1000$ and optimize the bounds for $T = \lambda(\mathbf{p}) \cdot U$ server steps, where $\lambda(\mathbf{p})$ is the average network speed.

We choose the sampling probabilities $\mathbf{p}$ and the step size $\eta$ by solving the constrained optimization problem $\min_{\mathbf{p},\eta} G(\mathbf{p},\eta)$ as a function of $\eta \leq \eta_{\max}(\mathbf{p})$, where
\begin{equation}
G(\mathbf{p},\eta)= \frac{A}{\eta (\lambda(\mathbf{p}) U+1)} + \frac{\eta L B}{n}\sum\nolimits_{i=1}^n \frac{1}{np_i} + \frac{\eta^2 L^2 B \cardgrad}{n} \sum\nolimits_{i=1}^n  \frac{\sum\nolimits_{k=0}^{\lambda(\mathbf{p}) U} \m}{n p_i^2(\lambda(\mathbf{p}) U+1)}\eqsp,
\end{equation}
and where $A=\PE[f(\mu_0) - f(\mu_{\nsteps+1})]$. In \Cref{fig:bounds_concu_time} we run the same simulation as in \Cref{sec:optim_bounds}, for a fixed amount of time $U$.
\begin{figure}
    \centering
    \includegraphics[scale=.6]{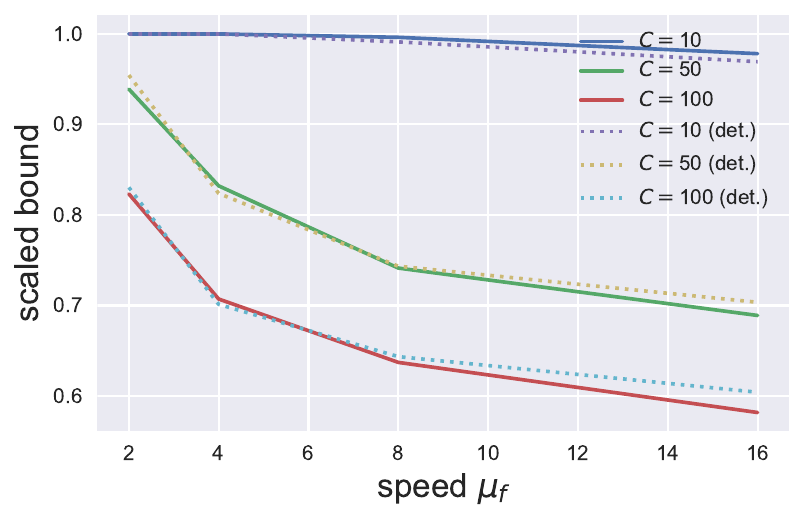}
    \caption{Relative improvements of the upper bounds as a function of the speed for different concurrency levels.}
    \label{fig:bounds_concu_time}
\end{figure}
Taking into account a fixed amount of time $U$ instead of CS epochs $\nsteps$ also favours our approach. The experimental results suggest to sample less fast nodes. It reduces delays (in number of steps), but increases the average time spent between two consecutive server steps. This trade-off is key for optimizing the bounds. When the concurrency is small (\wrt\ $n$), uniform sampling appears as the best strategy. However, by taking $p=8.5 \cdot 10^{-3}$, for full concurrency ($\cardgrad=n$), the bound can be reduced by $40\%$.
\section{2 clusters under saturation}
\label{sec:app:2clusters}
\subsection{Example with 2 saturated clusters}
In \Cref{sec:closed_network}, we propose a study of the delays and queue lengths when the number of task goes to infinity with a rate controlled by some $\iota>0$. In particular, we introduce the scaled intensities for slow and fast nodes as:
\begin{equation}
    \begin{cases}
        \gamma_{s}(\iota) = \frac{\max_{i \in [1, n]}(\theta_i)}{\theta_s} = \frac{\theta_s}{\theta_s} =1\\
        \gamma_{f}(\iota) = \frac{\max_{i \in [1, n]}(\theta_i)}{\theta_f} = \frac{\theta_s}{\theta_f} =1+\underbrace{c_f \cdot \iota^{\alpha-1}}_{\text{deviation from slow speed}}
    \end{cases}
\end{equation}
Note $c_f>0$ and $\alpha \leq 1$ are parameters we are free to choose to match the number of tasks in the network. In particular, the total number of tasks also scales as follows: $\beta l^{1-\alpha} = \cardgrad+1$.
Thanks to \Cref{prop:convSR}, we can bound the number of server steps when a task arrive and quits some node $i$ as:
\begin{equation}
\lim_{\iota \to \infty} \iota^{\alpha-1} \mi(\iota) \leq  \lim_{\iota \to \infty}{\lambda \over \mu_i}(\iota^{\alpha-1}\PE[X_{i}^\iota]+1),
\end{equation}
where $\lambda=n_f \mu_s + (n-n_f) \mu_s$.
Hence,
\begin{equation}
    \lim_{\iota \to \infty} \iota^{\alpha-1} \mi(\iota) \leq  \frac{n_f \mu_s + (n-n_f) \mu_s}{\mu_i}(\frac{1}{c_f} \Gamma(c_f \beta)+1)
\end{equation}
We will further assume $n_f=\frac{n}{2}$, and $p_i=\frac{1}{n}$.  Under this setting, we have $\Gamma(c_F \beta) \simeq 1$ and
\begin{align}
\frac{\iota^{1-\alpha}}{c_f} \Gamma(c_f \beta)&= \frac{1}{\gamma_f(\iota) -1 }\\
&= \frac{1}{\frac{\theta_s}{\theta_f} -1 }\\
&= \frac{1}{\frac{\mu_fp_s}{\mu_sp_f} -1}\\
&=\frac{1}{\frac{\mu_f}{\mu_s} -1}.
\end{align}
For fast nodes, the delay can be simplified as:
\begin{equation}
    \lim_{\iota \to \infty} \mi(\iota) \leq \frac{n}{2}\frac{\mu_s+\mu_f}{\mu_f}\frac{1}{\frac{\mu_f}{\mu_s} -1}.
\end{equation}
For slow nodes, this simplifies as:
\begin{equation}
    \lim_{\iota \to \infty} \mi(\iota) \leq \frac{n}{2}\frac{\mu_s+\mu_f}{\mu_s}(\frac{2}{n}\cardgrad-\frac{1}{\frac{\mu_f}{\mu_s} -1}).
\end{equation}
In the following we consider $n=10$ clients, split in two clusters of same size: fast nodes with rate $\mu_f = 1.2$, and slow nodes with rates $\mu_s = 1$. We simulate up to $\nsteps=10^6$ server steps, and plot the distribution of the delays (in number of server steps). We saturate the network with $\cardgrad=1000$ tasks.
Hence we can estimate the following:
\begin{equation}
    \begin{cases}
        \lim_{\iota \to \infty} \mi(\iota) \leq \frac{n}{\frac{\mu_f}{\mu_s} -1} \simeq 5 n, \quad \forall i \in [1, n_f],\\
        \lim_{\iota \to \infty} \mi(\iota) \leq (\frac{2\cardgrad}{n} - \frac{1}{\frac{\mu_f}{\mu_s} -1})n \simeq 195 n, \quad \forall i \in [n_f+1, n].
    \end{cases}
\end{equation}
All delays bounds estimations have a closed form in the 2-cluster saturated regime: they only depend on the number of tasks in the network $\cardgrad$, on the number of nodes $n$, and on the intensity of nodes $\mu_f, \mu_s$.

\begin{figure}[h!]
    \centering
\includegraphics[scale=.8]{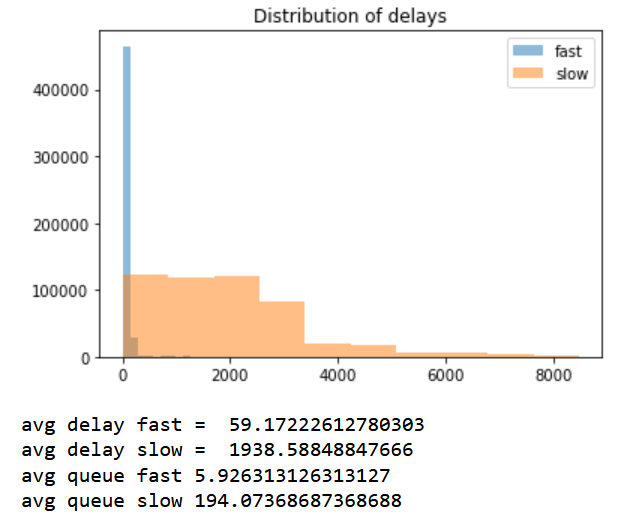}
    \caption{Histogram of fast and slow delays (in number of server steps) for a uniform sampling scheme.}
    \label{fig:app_simu_delays2clusters}
\end{figure}
Our numerical experiment in \Cref{fig:simu_delays2clusters} gives an average delay of $59 \sim 5n$ for fast nodes. The average delay for slow nodes reaches the value $1938 \simeq 195n$. And the queue lengths also correspond to the expected values. It is also important to point out that the average delays are way smaller than the maximum delay experienced in the $K=10^6$ steps. This further highlights the necessity to switch from analysis that depend on the $\tau_{max}$ quantity, to our analysis that only depends on the expected delays.
\subsection{Optimal sampling strategy under saturation}
In the previous paragraph we kept the sampling probability $p_i$ uniform.
The previous computation gives
\begin{align}
\frac{l^{1-\alpha}}{c_f} \Gamma(c_f \beta)&= \frac{1}{\gamma_f^l -1 }\\
&= \frac{1}{\frac{\theta_s}{\theta_f} -1 }\\
&= \frac{1}{\frac{\mu_fp_s}{\mu_sp_f} -1}.
\end{align}
Sticking to the previous assumptions, we want to minimize the quantity
\begin{multline}
G(\mathbf{p},\eta)= \frac{A}{\eta (\nsteps+1)} + \frac{\eta L B}{n}(\sum\nolimits_{i=1}^{n_f} \frac{1}{np}+\sum\nolimits_{i=n_f+1}^{n} \frac{1}{n(\frac{1}{n_-n_f}-p\frac{n_f}{n-n_f})^{-1}}) + \frac{\eta^2 L^2 B \cardgrad}{n}( \sum\nolimits_{i=1}^{n_f}  \frac{m_f}{n p^2} \\
+\sum\nolimits_{i=n_f+1}^n  \frac{m_s}{n (\frac{1}{n_-n_f}-p\frac{n_f}{n-n_f})^{-2}})\eqsp,
\end{multline}
This is equivalent to minimizing:
\begin{multline}
G(\mathbf{p},\eta)= \frac{A}{\eta (\nsteps+1)} + \frac{\eta L B}{n}(\sum\nolimits_{i=1}^{n_f} \frac{1}{np}+\sum\nolimits_{i=n_f+1}^{n} \frac{1}{n(\frac{1}{n_-n_f}-p\frac{n_f}{n-n_f})^{-1}}) + \frac{\eta^2 L^2 B \cardgrad}{n}( \sum\nolimits_{i=1}^{n_f}  \frac{\frac{n_f \mu_s + (n-n_f) \mu_s}{\mu_f}\frac{1}{\frac{\mu_fp_s}{\mu_sp_f} -1}}{n p^2} \\
+\sum\nolimits_{i=n_f+1}^n  \frac{\frac{n_f \mu_s + (n-n_f) \mu_s}{\mu_s}(\frac{\cardgrad}{n-n_f}-\frac{n_f}{n-n_f}\frac{1}{\frac{\mu_fp_s}{\mu_sp_f} -1})}{n (\frac{1}{n-n_f}-p\frac{n_f}{n-n_f})^{-2}})\eqsp,
\end{multline}
Hence we want to find $(\mathbf{p},\eta)$ that minimize the following:
\begin{multline}
G(\mathbf{p},\eta)= \frac{A}{\eta (\nsteps+1)} + \frac{\eta L B}{n}(\sum\nolimits_{i=1}^{n_f} \frac{1}{np}+\sum\nolimits_{i=n_f+1}^{n} \frac{1}{n(\frac{1}{n_-n_f}-p\frac{n_f}{n-n_f})^{-1}}) \\
+ \frac{\eta^2 L^2 B \cardgrad}{n}( \sum\nolimits_{i=1}^{n_f}  \frac{\frac{n_f \mu_s + (n-n_f) \mu_s}{\mu_f}\frac{1}{\frac{\mu_f}{\mu_s}(\frac{1}{p(n-n_f)}-\frac{n_f}{n-n_f}) -1}}{n p^2} \\
+\sum\nolimits_{i=n_f+1}^n  \frac{\frac{n_f \mu_s + (n-n_f) \mu_s}{\mu_s}(\frac{\cardgrad}{n-n_f}-\frac{n_f}{n-n_f}\frac{1}{\frac{\mu_f}{\mu_s}(\frac{1}{p(n-n_f)}-\frac{n_f}{n-n_f}) -1})}{n (\frac{1}{n-n_f}-p\frac{n_f}{n-n_f})^{-2}})\eqsp,
\end{multline}

The uniform sampling strategy ($p=\frac{1}{n}=0.1$) and higher probability values, give larger average delays.But very small sampling probabilities lead to a sharp increase in the delays. It is not easy to find a closed form formula of the optimal probability value $\mathbf{p}$. Our simulations suggest an optimal value of $p = 7.5 \cdot 10^{-3}$.

In \Cref{fig:app_histo_opt}, we have run the same simulations as in \Cref{fig:simu_delays2clusters}, except that we do not sample nodes uniformly at random. Instead, we sample fast nodes with a probability $p=7.5 \cdot 10^{-3}$, and slow nodes with a probability $\frac{2}{n}-p$. Our simulations shows the average delay is divided by $10$ and $2$, for fast and slow nodes respectively.
\begin{figure}[h!]
    \centering
    \includegraphics[scale=.8]{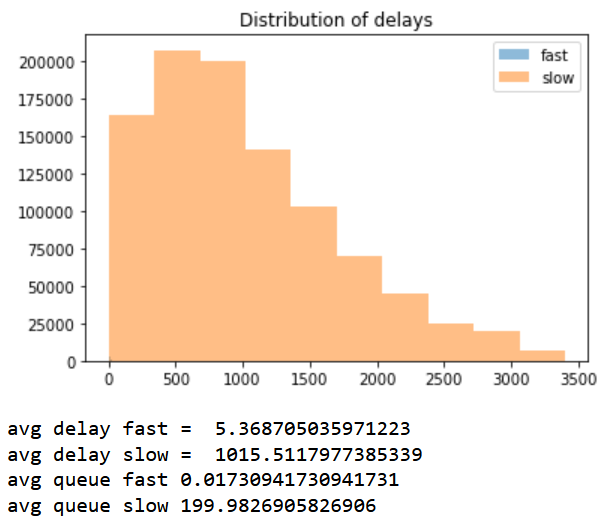}
    \caption{Histogram of fast and slow delays (in number of server steps) for an optimal sampling strategy.}
    \label{fig:app_histo_opt}
\end{figure}
\section{3 clusters scaling regime under saturation}
\label{sec:app:3clusters}
We consider three clusters of nodes of size $n_f$, $n_m-1-n_f$, and $n-1-n_m$, respectively. Nodes $\forall i \leq n_f$ are considered as fast, whereas nodes $\forall i > n_m$ are considered as slow (and will likely get more saturated).

The medium nodes ($\forall i \in [n_f+1, n_m]$) have an intermediate computational speed. We assume nodes from the same cluster have the same intensity $\mu_f, \mu_m, \mu_s$, for fast, medium, and slow nodes respectively. For practical reason we assume now that nodes $\forall i > n_m$ are the slowest ones, and has an intensity $\theta_s > \theta_j, \forall j <n$. This assumption is not restrictive due to the close nature of the network, and allows us to simplify the problem by splitting nodes into clusters.

This results in the \emph{scaled} intensities $\gamma_f(\iota), \gamma_m(\iota), \gamma_s(\iota)$,
where $\gamma_s(\iota)= 1 $, $\gamma_m(\iota) = 1+ c_m \iota^{\alpha-1}$, and $\gamma_f(\iota) = 1+ c_f\iota^{\delta-1}$; with $\alpha \leq 1$ and $\delta > 1$. The constant task constraint translates into the existence of $\beta$ such that $\beta \iota^{1-\alpha} = \cardgrad+1$. The particular choice of $\alpha \leq 1$ in \cite{van2021scaling} allows us to obtain traffic loads of nodes that tend to $1$ as $\iota \rightarrow \infty$, and we could directly apply Corollary 2 from \cite{van2021scaling}. But this setting is inconsistent with the practical Federated Learning framework we consider in this section: in practice most of fast clients have an empty queue. Hence, we stick to $\delta > 1$, and we can apply the results of \cite[Corollary.3,][]{van2021scaling}. This work gives a precise results on the queue length of saturated nodes, in the limit of high traffic loads. The queue length of the remaining queue are defined by the population size constraint. Note because the unnormalized queue length of fast nodes converges to a finite-mean random variable, there is no need to scale it.
\begin{proposition}[Corollary.3 in \cite{van2021scaling}]
  In stationary regime,  as $\iota \to \infty$, $X_{i}^\iota, \forall i \in [1, n_f]$ become degenerate with value $0$, and
  \begin{equation}
       c_m \iota^{\alpha-1} {X_{i}^\iota} \rightarrow_{d.}  \CPE{E_i}{\sum\nolimits_{j=n_f+1}^{n_m} \frac{E_j}{c_m} \leq \beta}, \forall i \in [n_f+1, n_m],
\end{equation}
with $E_i$ unit mean exponential distributions.
\end{proposition}
As a consequence, using as before dominated convergence we can estimate the following expected value (expected stationary queue lengths of fast, medium, and slow nodes respectively):
\begin{equation}
    \begin{cases}
        \lim_{\iota \to \infty} \PE[X^\iota_i] = 0 , \quad \forall i \in [1, n_f], \\
        \lim_{\iota \to \infty} \iota^{\alpha-1} \PE[X^\iota_i] = \frac{1}{c_m} \Gamma(c_m \beta), \quad \forall i \in [n_f+1,n_m],\\
        \lim_{\iota \to \infty} \iota^{\alpha-1} \PE[X^\iota_i] = \frac{1}{n-n_m} \left (\beta - (n_m-n_f)\frac{1}{c_m} \Gamma(c_m \beta) \right ), \quad \forall i \in [n_m+1,n].
    \end{cases}
\end{equation}
Hence, we can estimate the number of server steps when a task arrive and quits some node $i$ as :
\begin{equation}
\lim_{\iota \to \infty} \iota^{\alpha-1} \mi(\iota) \leq  \lim_{\iota \to \infty}{\lambda \over \mu_i}(\iota^{\alpha-1}\PE[X_{i}^\iota]+1),
\end{equation}
where $\lambda=n_f \mathbb{P}(X_f>0)\mu_f + (n_m-n_f) \mu_m + (n-n_m) \mu_s$ (because fast nodes have \emph{almost} empty queue $X_f$ in the considered stationary setting).

We will further assume $n_f=\frac{n}{3}$, $n_m=\frac{2n}{3}$, and $p=\frac{1}{n}$. Under these conditions, we have $\Gamma(c_m \beta) = \Gamma(\cardgrad(\frac{\mu_m}{\mu_s}-1)) \simeq 1$.
For fast nodes, the delay can be simplified as:
\begin{equation}
    \lim_{\iota \to \infty} \mi(\iota) \leq \frac{\lambda}{\mu_f}.
\end{equation}
For medium nodes, the delay can be simplified as:
\begin{equation}
    \lim_{\iota \to \infty} \mi(\iota) \leq \frac{\lambda}{\mu_m}\frac{1}{\frac{\mu_m}{\mu_s} -1}.
\end{equation}
For slow nodes, this simplifies as:
\begin{equation}
    \lim_{\iota \to \infty} \mi(\iota) \leq \frac{\lambda}{\mu_s}(\frac{3}{n}\cardgrad-\frac{1}{\frac{\mu_m}{\mu_s} -1}).
\end{equation}
In the following we consider $n=9$ clients, split in three clusters of same size: fast nodes with rate $\mu_f=10$, medium nodes with rate $\mu_m=1.2$, and slow nodes with rate $\mu_s=1$. We simulate up to $\nsteps=10^6$ server steps, and plot the distribution of the delays (in number of server steps). Under this setting, we have $\frac{\iota^{1-\alpha}}{c_m} \Gamma(c_m \beta) = \frac{1}{\frac{\mu_m}{\mu_s} -1} = 5$. Hence we can estimate the following:
\begin{equation}
    \begin{cases}
        \lim_{\iota \to \infty} \mi(\iota) \leq \frac{n_f \mathbb{P}(X_f>0)\mu_f + (n_m-n_f) \mu_m + (n-n_m) \mu_s}{\mu_f} \simeq 3\mathbb{P}(X_f>0), \quad \forall i \in [1, n_f],\\
        \lim_{\iota \to \infty} \mi(\iota) \leq \frac{n_f \mathbb{P}(X_f>0)\mu_f + (n_m-n_f) \mu_m + (n-n_m) \mu_s}{\mu_m}\frac{1}{\frac{\mu_m}{\mu_s} -1}, \quad \forall i \in [n_f+1, n_m],\\
        \lim_{\iota \to \infty} \mi(\iota) \leq \frac{n_f \mathbb{P}(X_f>0)\mu_f + (n_m-n_f) \mu_m + (n-n_m) \mu_s}{\mu_s} (\frac{3\cardgrad}{n} - \frac{1}{\frac{\mu_m}{\mu_s} -1}), \quad \forall i \in [n_m+1, n].
    \end{cases}
\end{equation}
\begin{figure}[h!]
    \centering
\includegraphics[scale=.8]{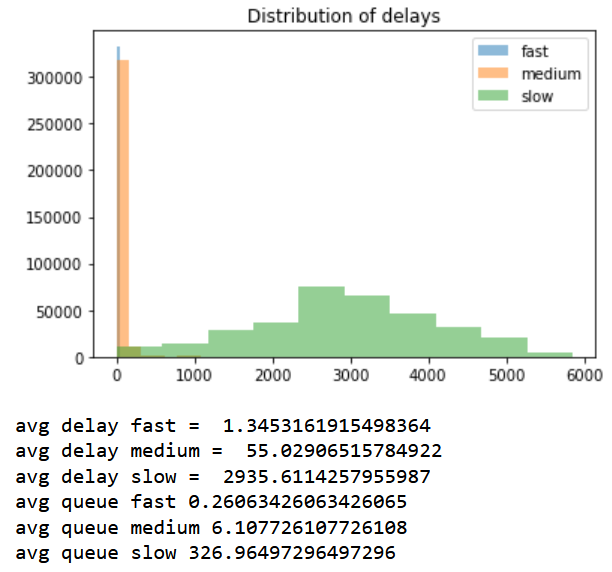}
    \caption{We assign $\cardgrad=1000$ tasks to a network of $n=9$ nodes split in 3 clusters.}
    \label{fig:simu_delays}
\end{figure}
The simulation gives $\lambda \simeq 9$, and we recover the theoretical delays: the average delay for fast nodes is close to $1$, the average delay for medium node is $55 \simeq 5 \frac{\lambda}{\mu_m}$, and the average delay for slow nodes is about $2935 \simeq 325 \frac{\lambda}{\mu_s}$.

\section{Deep learning experiments details}
\label{sec:app:deep_exp}
\begin{table}[]
    \centering
    \begin{tabular}{c|c|c|c}
        Method  & FedBuff & AsyncSGD & \Algo\ \\
        \hline
        Accuracy on the CS test set & $49.89 \pm 0.77$ & $59.09 \pm 1.97$  & $66.61 \pm 3.26$
    \end{tabular}
    \caption{Performance average (mean $\pm$ std) over $10$ random seeds for the CIFAR-10 task.}
    \label{tab:avg_accuracy}
\end{table}
\subsection{Simulation}
\label{sec:app:simuruntime}
We based our simulations mainly on the code developed by \cite{nguyen2022federated}: we assume a server and $n$ clients, each of which initially has a unique split of the training dataset. To adequately capture the time spent on the server side for computations and orchestration of centralized learning, two quantities are implemented: the server waiting time (the time the server waits between two consecutive calls ) and the server interaction time (the time the server takes to send and receive the required data). In all experiments, they are set to $4$ and $3$, respectively. When a client $i$ receives a new task, we take a new sample from an exponential distribution (with mean $\frac{1}{\mu_i}$, where $\mu_i$ is the rate of node $i$), and stack the gradient computation on top of the client queue.
\subsection{Implementation}
\label{sec:app:concu_works}
In \Cref{sec:experiments} we have simulated experiments and run the code for the concurrent approaches \texttt{AsyncSGD} and \texttt{FedBuff}. We also propose an implementation of \texttt{FedAvg} in the supplemental material. \texttt{FedAvg} is a standard synchronous method. At the beginning of each round, the central node $s$ selects clients uniformly at random and broadcast its current model. Each of these clients take the central server value and then performs exactly $K$ local steps, and then sends the resulting model progress back to the server. The server then computes the average of the $s$ received models and updates its model. In this synchronous structure, the server must wait in each round for the slowest client to complete its update.

\texttt{AsyncSGD} is an asynchronous method that initially randomly selects $\cardgrad$ clients. Then, a server step is done when a new task is completed and sent back to the server. The server uniformly selects a new client and send a new task. While \texttt{AsyncSGD} was tested on a simple task in \cite{koloskova2022sharper}, we have developed a deep learning version of the algorithm (see supplemental material) based on a list of dictionaries (to simulate a network of waiting queues).

For each global step, in \texttt{FedBuff}, the runtime is the sum of the server interaction time and the time spent feeding the buffer of size $Z$. The waiting time for feeding the buffer depends on the respective local runtimes of the slow and fast clients, as well as on the ratio between slow and fast clients: in the code, we reset a counter at the beginning of each global step and read the runtime when the $Z^{th}$ local update arrives. In \texttt{AsyncSGD} and \Algo, the runtime is defined by the closed Jackson network properties. \end{document}